\documentclass[screen,libertine]{techreport}
\providecommand{\keywords}[1]{ \small\textbf{\textit{Keywords---}} #1}

\AtEndPreamble{%
  \theoremstyle{trdefinition}%
  \newtheorem{counterexample}[theorem]{Counterexample}%
  \theoremstyle{trplain}%
}
\crefname{counterexample}{Ex.}{Exs.}

\renewcommand{\defeq}{\coloneqq}
\newcommand{\m}[1]{\mathsf{#1}}

\newcommand{\prob}{\ensuremath{\bbP}}
\newcommand{\expe}{\ensuremath{\bbE}}

\newcommand{\lang}{\textsc{Appl}}
\newcommand{\reason}[1]{\dagger~\text{#1}~\dagger}

\newcommand{\bind}{\mathbin{\gg\!=}}

\newcommand{\iskip}{\kw{skip}}
\newcommand{\itick}[1]{\kw{tick}(#1)}
\newcommand{\iassign}[2]{#1 \coloneqq #2}
\newcommand{\isample}[2]{#1 \sim #2}
\newcommand{\iprob}[3]{\kw{if}~\kw{prob}(#1)~\kw{then}~#2~\kw{else}~#3~\kw{fi}}
\newcommand{\icond}[3]{\kw{if}~#1~\kw{then}~#2~\kw{else}~#3~\kw{fi}}
\newcommand{\iloop}[2]{\kw{while}~#1~\kw{do}~#2~\kw{od}}
\newcommand{\iinvoke}[1]{\kw{call}~#1}

\newcommand{\ctop}{\kw{tt}}
\newcommand{\cneg}[1]{\kw{not}~#1}
\newcommand{\cconj}[2]{#1~\kw{and}~#2}
\newcommand{\cbin}[3]{#2 \mathrel{#1} #3}
\newcommand{\cle}[2]{\cbin{\le}{#1}{#2}}

\newcommand{\eadd}[2]{#1 + #2}
\newcommand{\esub}[2]{#1 - #2}
\newcommand{\emul}[2]{#1 \times #2}

\newcommand{\kstop}{\kw{Kstop}}
\newcommand{\kseq}[2]{\kw{Kseq}~#1~#2}
\newcommand{\kloop}[3]{\kw{Kloop}~#1~#2~#3}

\begin{document}
\title{Expected-Cost Analysis for Probabilistic Programs and Semantics-Level Adaption of Optional Stopping Theorems}
\author[1]{Di Wang}
\author[1]{Jan Hoffmann}
\author[2]{Thomas Reps}
\affil[1]{Carnegie Mellon University}
\affil[2]{University of Wisconsin}
\date{}
\maketitle

%



\section{A Probabilistic Programming Language}
\label{Sec:PPL}

This article uses an imperative arithmetic probabilistic programming language \lang{} that supports general recursion and continuous distributions, where program variables are real-valued.
We use the following notational conventions.
Natural numbers $\bbN$ \emph{exclude} $0$, i.e., $\bbN \defeq \{ 1,2,3,\cdots\} \subseteq \bbZ^+ \defeq \{0,1,2,\cdots\}$.
The \emph{Iverson brackets} $[\cdot]$ are defined by $[\varphi]=1$ if $\varphi$ is true and otherwise $[\varphi]=0$.
We denote updating an existing binding of $x$ in a finite map $f$ to $v$ by $f[x \mapsto v]$.
Interested readers can refer to textbooks in the literature~\cite{book:Billingsley12,book:Williams91} for details about measure theory.

\subsection{Preliminaries: Measure Theory}
\label{Sec:MeasureTheory}

A \emph{measurable space} is a pair $(S,\calS)$, where $S$ is a nonempty set, and $\calS$ is a $\sigma$-algebra on $S$, i.e., a family of subsets of $S$ that contains $\emptyset$ and is closed under complement and countable unions.
The smallest $\sigma$-algebra that contains a family $\calA$ of subsets of $S$ is said to be \emph{generated} by $\calA$, denoted by $\sigma(\calA)$.
Every topological space $(S,\tau)$ admits a \emph{Borel $\sigma$-algebra}, given by $\sigma(\tau)$.
This gives canonical $\sigma$-algebras on $\bbR$, $\bbQ$, $\bbN$, etc.
A function $f : S \to T$, where $(S,\calS)$ and $(T,\calT)$ are measurable spaces, is said to be \emph{$(\calS,\calT)$-measurable}, if $f^{-1}(B) \in \calS$ for each $B \in \calT$.
If $T=\bbR$, we tacitly assume that the Borel $\sigma$-algebra is defined on $T$, and we simply call $f$ \emph{measurable}, or a \emph{random variable}.
Measurable functions form a vector space, and products, maxima, and limiting operations preserve measurability.\footnote{For limiting operations to be well-defined, we consider the \emph{extended} real number line $\overline{\bbR} \defeq \bbR \cup \{\infty,-\infty\}$.}

A \emph{measure} $\mu$ on a measurable space $(S,\calS)$ is a mapping from $\calS$ to $[0,\infty]$ such that
(i) $\mu(\emptyset)=0$, and
(ii) for all pairwise-disjoint $\{A_n\}_{n \in \bbN}$ in $\calS$, it holds that $\mu(\bigcup_{n \in \bbN} A_n) = \sum_{n \in \bbN} \mu(A_n)$.
The triple $(S,\calS,\mu)$ is called a \emph{measure space}.
A measure $\mu$ is called a \emph{probability} measure, if $\mu(S)=1$.
We denote the collection of probability measures on $(S,\calS)$ by $\bbD(S,\calS)$.
The \emph{zero measure} $\mathbf{0}$ is defined as $\lambda A. 0$.
For each $x \in S$, the \emph{Dirac measure} $\delta(x)$ is defined as $\lambda A. [x \in A]$.
For measures $\mu$ and $\nu$, we write $\mu+\nu$ for the measure $\lambda A. \mu(A) + \nu(A)$.
For measure $\mu$ and scalar $c \ge 0$, we write $c \cdot \mu$ for the measure $\lambda A. c \cdot \mu(A)$.

The \emph{integral} of a measurable function $f$ on $A \in \calS$ with respect to a measure $\mu$ on $(S,\calS)$ is defined following Lebesgue's theory and is denoted by $\mu(f;A)$, $\int_A f d\mu$, or $\int_A f(x) \mu(dx)$.
If $\mu$ is a probability measure, we call the integral as the \emph{expectation} of $f$, written $\expe_{x \sim \mu}[f;A]$, or simply $\expe[f;A]$ when the scope is clear in the context.
If $A = S$, we tacitly omit $A$ from the notations.
For each $A \in \calS$, it holds that $\mu(f;A) = \mu(f \mathrm{I}_A)$, where $\mathrm{I}_A$ is the indicator function for $A$.
If $f$ is nonnegative, then $\mu(f)$ is well-defined with the understanding that the integral can be infinite.
If $\mu(|f|) < \infty$, then $f$ is said to be \emph{integrable}, written $f \in \calL^1(S,\calS,\mu)$, and its integral is well-defined.
Integration is \emph{linear}, in the sense that for any $c,d \in \bbR$ and integrable functions $f,g$, $c \cdot f + d \cdot g$ is integrable and $\mu(c \cdot f + d \cdot g) = c \cdot \mu(f) + d \cdot \mu(g)$.

A \emph{kernel} from a measurable space $(S,\calS)$ to another $(T,\calT)$ is a mapping from $S \times \calT$ to $[0,\infty]$ such that
(i) for each $x \in S$, the set function $\lambda B. \kappa(x,B)$ is a measure on $(T,\calT)$, and
(ii) for each $B \in \calT$, the function $\lambda x. \kappa(x,B)$ is measurable.
We write $\kappa : (S,\calS) \rightsquigarrow (T,\calT)$ to declare that $\kappa$ is a kernel from $(S,\calS)$ to $(T,\calT)$.
Intuitively, kernels describe measure transformers from one measurable space to another.
A kernel $\kappa$ is called a \emph{probability} kernel, if $\kappa(x,T)=1$ for all $x \in S$.
We denote the collection of probability kernels from $(S,\calS)$ to $(T,\calT)$ by $\bbK((S,\calS),(T,\calT))$.
If the two measurable spaces coincide, we simply write $\bbK(S,\calS)$.
We can \emph{push-forward} a measure $\mu$ on $(S,\calS)$ to a measure on $(T,\calT)$ through a kernel $\kappa : (S,\calS) \rightsquigarrow (T,\calT)$ by integration: $\mu \bind \kappa \defeq \lambda B. \int_{S} \kappa(x,B)\mu(dx)$.

We review two important convergence theorems for series of random variables.

\begin{theorem}[Monotone convergence theorem]\label{The:MON}
  If $\{f_n\}_{n \in \bbN}$ is a non-decreasing sequence of nonnegative measurable functions on a measure space $(S,\calS,\mu)$, and $\{f_n\}_{n \in \bbN}$ converges to $f$ pointwise, then $\lim_{n \to \infty} \mu(f_n) = \mu(f) \le \infty$.
  
  Further, the theorem still holds if $f$ is chosen as a measurable function and ``$\{f_n\}_{n \in \bbN}$ converges to $f$ pointwise'' holds \emph{almost everywhere}, rather than \emph{everywhere}.
\end{theorem}

\begin{theorem}[Dominated convergence theorem]\label{The:DOM}
  If $\{f_n\}_{n \in \bbN}$ is a sequence of measurable functions on a measure space $(S,\calS,\mu)$, $\{f_n\}_{n \in \bbN}$ converges to $f$ pointwise, and $\{f_n\}_{n \in \bbN}$ is dominated by a nonnegative integrable function $g$ (i.e., $|f_n(x)| \le g(x)$ for all $n \in \bbN, x \in S$), then $f$ is integrable and $\lim_{n \to \infty} \mu(f_n) = \mu(f)$.
  
  Further, the theorem still holds if $f$ is chosen as a measurable function and ``$\{f_n\}_{n \in \bbN}$ converges to $f$ pointwise and is dominated by $g$'' holds \emph{almost everywhere}, rather than \emph{everywhere}.
\end{theorem}

\subsection{Syntax and Semantics of Probabilistic Programs}
\label{Sec:SyntaxAndSemantics}

\begin{figure}
\centering
\begin{align*}
  S & \Coloneqq \iskip \mid \itick{c} \mid \iassign{x}{E} \mid \isample{x}{D} \mid \iinvoke{f} \mid \iloop{L}{S}  \\
  & \mid \iprob{p}{S_1}{S_2} \mid \icond{L}{S_1}{S_2}  \mid S_1 ; S_2 \\
  L & \Coloneqq \ctop \mid \cneg{L} \mid \cconj{L_1}{L_2} \mid \cle{E_1}{E_2} \\
  E & \Coloneqq x \mid c \mid \eadd{E_1}{E_2} \mid \emul{E_1}{E_2} \\
  D & \Coloneqq \kw{uniform}(a,b) \mid \cdots
\end{align*}
\caption{Syntax of \lang{}, where $p \in [0,1]$, $a,b,c \in \bbR$, $a < b$, $x\in \mathsf{VID}$ is a variable, and $f \in \mathsf{FID}$ is a function identifier.}
\label{Fi:Syntax}
\end{figure}

\paragraph{Syntax.}

\cref{Fi:Syntax} presents the syntax of \lang{}, where the metavariables $S$, $L$, $E$, and $D$ stand for statements, conditions, expressions, and distributions, respectively.
Each distribution $D$ is associated with a probability measure $\mu_D \in \bbD(\bbR)$.
For example, $\kw{uniform}(a,b)$ describes a uniform distribution on the interval $[a,b]$, and its corresponding probability measure is the integration of its density function $\mu_{\kw{uniform}(a,b)}(A) \defeq \int_A \frac{[a \le x \le b]}{b - a} dx $.
The statement ``$\isample{x}{D}$'' is a \emph{random-sampling} assignment, which draws from the distribution $\mu_D$ to obtain a sample value and then assigns it to $x$.
The statement ``$\iprob{p}{S_1}{S_2}$'' is a \emph{probabilistic-branching} statement, which executes $S_1$ with probability $p$, or $S_2$ with probability $(1-p)$.

The statement ``$\iinvoke{f}$'' makes a (possibly recursive) call to the function with identifier $f \in \mathsf{FID}$.
In this article, we assume that the functions only manipulate states that consist of global program variables.
The statement $\itick{c}$, where $c \in \bbR$ is a constant, is used to define the \emph{cost model}.
It adds $c$ to an anonymous global cost accumulator.
Note that our implementation supports local variables, function parameters, return statements, as well as accumulation of non-constant costs;
the restrictions imposed here are not essential, and are introduced solely to simplify the presentation.

We use a pair $\tuple{\scrD, S_{\mathsf{main}}}$ to represent an \lang{} program, where $\scrD$ is a finite map from function identifiers to their bodies and $S_{\mathsf{main}}$ is the body of the main function.

\begin{figure*}
\flushleft{\small\fbox{$\gamma \vdash E \Downarrow r $}\;\;\;\;``the expression $E$ evaluates to a real value $r$ under the valuation $\gamma$''}
\begin{mathpar}\small
  \Rule{E-Var}{ \gamma(x) = r }{ \gamma \vdash x \Downarrow r }
  \and
  \Rule{E-Const}{ }{ \gamma \vdash c \Downarrow c }
  \and
  \Rule{E-Add}{ \gamma \vdash E_1 \Downarrow r_1 \\ \gamma \vdash E_2 \Downarrow r_2 \\ r = r_1 + r_2 }{ \gamma \vdash \eadd{ E_1 }{ E_2} \Downarrow r }
  \and
  \Rule{E-Mul}{ \gamma \vdash E_1 \Downarrow r_1 \\ \gamma \vdash E_2 \Downarrow r_2 \\ r = r_1 \cdot r_2}{\gamma \vdash \emul{E_1}{ E_2 }\Downarrow r}
\end{mathpar}
\flushleft{\small\fbox{$\gamma \vdash L \Downarrow b$}\;\;\;\;``the condition $L$ evaluates to a Boolean value $b$ under the valuation $\gamma$''}
\begin{mathpar}\small
  \Rule{E-Top}{ }{ \gamma \vdash \ctop \Downarrow \top }
  \and
  \Rule{E-Neg}{ \gamma \vdash L \Downarrow b }{ \gamma \vdash \cneg{L} \Downarrow \neg b }
  \and
  \Rule{E-Conj}{ \gamma \vdash L_1 \Downarrow b_1 \\ \gamma \vdash L_2 \Downarrow b_2 }{ \gamma \vdash \cconj{L_1}{L_2} \Downarrow b_1 \wedge b_2 }
  \and
  \Rule{E-Le}{ \gamma \vdash E_1 \Downarrow r_1 \\ \gamma \vdash E_2 \Downarrow r_2 }{ \gamma \vdash \cle{E_1}{E_2} \Downarrow [r_1 \le r_2] }
\end{mathpar}
\flushleft{\small\fbox{$\tuple{\gamma,S,K,\alpha} \mapsto \mu$}\;\;\;\;``the configuration $\tuple{\gamma,S,K,\alpha}$ steps to a probability distribution $\mu$ on $\tuple{\gamma',S',K',\alpha'}$'s''}
\begin{mathpar}\small
  \Rule{E-Skip-Stop}{ }{ \tuple{\gamma,\iskip,\kstop,\alpha} \mapsto \delta(\tuple{\gamma,\iskip,\kstop,\alpha}) }
  \and
  \Rule{E-Skip-Loop}{ \gamma \vdash L \Downarrow b }{ \tuple{\gamma, \iskip , \kloop{S}{L}{K} ,\alpha} \mapsto [b] \cdot \delta(\tuple{\gamma,S, \kloop{S}{L}{K} ,\alpha}) + [\neg b] \cdot \delta( \tuple{\gamma,\iskip, K,\alpha} ) }
  \and
  \Rule{E-Skip-Seq}{ }{ \tuple{\gamma,\iskip,\kseq{S}{K},\alpha} \mapsto \delta( \tuple{\gamma,S,K,\alpha} )  }
  \and
  \Rule{E-Tick}{  }{ \tuple{\gamma,\itick{c}  ,K,\alpha} \mapsto \delta(\tuple{\gamma, \iskip ,K,\alpha+c} ) }
  \and
  \Rule{E-Assign}{ \gamma \vdash E \Downarrow r }{ \tuple{\gamma, \iassign{x}{E} ,K,\alpha} \mapsto \delta( \tuple{\gamma[x \mapsto r], \iskip ,K,\alpha} ) }
  \and
  \Rule{E-Sample}{  }{ \tuple{\gamma,\isample{x}{D} ,K,\alpha} \mapsto \mu_D \bind \lambda r. \delta( \tuple{ \gamma[x \mapsto r],  \iskip ,K,\alpha } ) }
  \and
  \Rule{E-Call}{  }{ \tuple{\gamma, \iinvoke{f}  ,K,\alpha} \mapsto \delta( \tuple{\gamma , \scrD(f) , K,\alpha } ) }
  \and
  \Rule{E-Prob}{ }{ \tuple{\gamma, \iprob{p}{S_1}{S_2} ,K,\alpha} \mapsto p \cdot \delta(\tuple{\gamma,S_1,K,\alpha}) + (1-p) \cdot \delta(\tuple{\gamma,S_2,K,\alpha}) }
  \and
  \Rule{E-Cond}{  \gamma \vdash L \Downarrow b }{ \tuple{\gamma,\icond{L}{S_1}{S_2} ,K,\alpha} \mapsto [b] \cdot \delta(\tuple{\gamma,S_1,K,\alpha}) + [\neg b] \cdot \delta(\tuple{\gamma,S_2,K,\alpha}) }
  \and
  \Rule{E-Loop}{ }{ \tuple{\gamma,\iloop{L}{S},K,\alpha} \mapsto \delta( \tuple{ \gamma,\iskip,\kloop{L}{S}{K},\alpha } ) }
  \and
  \Rule{E-Seq}{  }{ \tuple{\gamma,S_1;S_2, K,\alpha } \mapsto \delta( \tuple{ \gamma,S_1,\kseq{S_2}{K},\alpha } ) }
\end{mathpar}  
\caption{Rules of the operational semantics of \lang{}.}
\label{Fi:OperationalSemantics}
\end{figure*}

\paragraph{Semantics.}

We present a small-step operational semantics with continuations.
We follow a distribution-based approach~\cite{ICFP:BLG16,JCSS:Kozen81}
to define an operational cost semantics for \lang{}.
A \emph{program configuration} $\sigma \in \Sigma$ is a quadruple $\tuple{\gamma,S,K,\alpha}$ where $\gamma : \mathsf{VID} \to \bbR$ is a program state that maps variables to values, $S$ is the statement being executed, $K$ is a continuation that described what remains to be done after the execution of $S$, and $\alpha \in \bbR$ is the global cost accumulator.
A \emph{continuation} $K$ is either an empty continuation $\kstop$, a loop continuation $\kloop{L}{S}{K}$, or a sequence continuation $\kseq{S}{K}$.
Note that there does not exist a continuation for function calls, because we assume that functions only manipulate global program variables.
Nevertheless, it is a common approach to include a continuation component in the program configurations if functions have local variables.
An execution of an \lang{} program $\tuple{\scrD,S_\mathsf{main}}$ is initialized with $\tuple{\lambda\_. 0, S_{\mathsf{main}}, \kstop,0}$, and the termination configurations have the form $\tuple{\_,\iskip,\kstop,\_}$.  

Different from a standard semantics where each program configuration steps to at most one new configuration, a probabilistic semantics may pick several different new configurations.
The evaluation relation for \lang{} has the form $\sigma \mapsto \mu$ where $\mu \in \bbD(\Sigma)$ is a probability measure over configurations.
\cref{Fi:OperationalSemantics} collects the evaluation rules.
Note that in \lang{}, expressions $E$ and conditions $L$ are deterministic, so we define a standard big-step evaluation relation for them, written $\gamma \vdash E \Downarrow r$ and $\gamma \vdash L \Downarrow b$, where $\gamma$ is a valuation, $r \in \bbR$, and $b \in \{ \top, \bot \}$.
Most of the rules, except \textsc{(E-Sample)} and \textsc{(E-Prob)}, are also deterministic as they step to a Dirac measure.
The rule \textsc{(E-Prob)} constructs a distribution whose support has exactly two elements, which stand for the two branches of the probabilistic choice.
We write $\delta(\sigma)$ for the \emph{Dirac measure} at $\sigma$, defined as $\lambda A.[\sigma\in A]$ where $A$ is a measurable subset of $\Sigma$.
We also write $p \cdot \mu_1 + (1-p)\cdot \mu_2$ for a convex combination of measures $\mu_1$ and $\mu_2$ where $p \in [0,1]$, defined as $\lambda A. p \cdot \mu_1(A) + (1-p) \cdot \mu_2(A)$.
The rule \textsc{(E-Sample)} \emph{pushes} the probability distribution of $D$ to a distribution over post-sampling program configurations.

\begin{example}\label{Exa:EvaluationRelation}
  Suppose that a random sampling statement is being executed, i.e., the current configuration is
  \begin{center}
  $\tuple{\{ t \mapsto t_0 \}, (\isample{t}{\kw{uniform}(-1,2)}),K_0,\alpha_0}.$
  \end{center}
  The probability measure for the uniform distribution is $\lambda A. \int_A \frac{[-1 \le x \le 2]}{3}   dx$.
  Thus by the rule \textsc{(E-Sample)}, we derive the post-sampling probability measure over configurations via the following density function:
  \begin{center}
  $
  \lambda \tuple{\gamma,S,K,\alpha}. \frac{[-1 \le \gamma(t) \le 2]}{3} \cdot [S = \iskip \wedge K = K_0 \wedge \alpha = \alpha_0].
  $
  \end{center}
\end{example}

\subsection{Meta-Theory}
\label{Sec:MetaTheory}

To formally construct a measurable space of program configurations,
our approach is to construct a measurable space for each of the four components of configurations, and then use their \emph{product} measurable space as the semantic domain.
The \emph{product} of two measurable spaces $(S,\calS)$ and $(T,\calT)$ is defined as $(S,\calS) \otimes (T,\calT) \defeq (S \times T, \calS \otimes \calT)$, where $\calS \otimes \calT$ is the smallest $\sigma$-algebra that makes coordinate maps measurable, i.e., $\sigma(\{ \pi_1^{-1}(A) : A \in \calS \} \cup \{ \pi_2^{-1}(B) : B \in \calT \}))$, where $\pi_i$ is the $i$-th coordinate map.
\begin{itemize}
  \item Valuations $\gamma : \mathsf{VID} \to \bbR$ are finite real-valued maps, so we define $(V,\calV) \defeq (\bbR^{\mathsf{VID}},\calB(\bbR^{\mathsf{VID}}))$ as the canonical structure on a finite-dimensional space.
  
  \item The executing statement $S$ can contain real numbers, so we need to lift the Borel $\sigma$-algebra on $\bbR$ to program statements.
  Intuitively, statements with exactly the same structure can be treated as vectors of parameters that correspond to their real-valued components.
  Formally, we achieve this by constructing a metric space on statements and then extracting a Borel $\sigma$-algebra from the metric space.
  \cref{Fi:MetricsForSemantics} presents an inductively defined metric $d_S$ on statements, as well as metrics $d_E$, $d_L$, and $d_D$ on expressions, conditions, and distributions, respectively, as they are required by $d_S$.
  We denote the result measurable space by $(S,\calS)$.
  
  \item
  Similarly, we construct a measurable space $(K,\calK)$ on continuations by extracting from a metric space.
  \cref{Fi:MetricsForSemantics} shows the definition of a metric $d_K$ on continuations.
  
  \item The cost accumulator $\alpha \in \bbR$ is a real number, so we define $(W,\calW) \defeq (\bbR,\calB(\bbR))$ as the canonical measurable space on $\bbR$.
\end{itemize}
Then the semantic domain is defined as the product measurable space of the four components: $(\Sigma,\calO) \defeq (V,\calV) \otimes (S,\calS) \otimes (K,\calK) \otimes (W,\calW)$.

\begin{figure*}
  \centering
  \small
  
  \noindent\hrulefill
  
  \begin{minipage}{1.0\textwidth}
  \begin{align*}
    d_E(x,x) & \defeq 0 \\
    d_E(c_1,c_2) & \defeq \abs{c_1-c_2} \\
    d_E(\eadd{E_{11}}{E_{12}},\eadd{E_{21}}{E_{22}}) & \defeq d_E(E_{11},E_{21}) + d_E(E_{12},E_{22}) \\
    d_E(\emul{E_{11}}{E_{12}},\emul{E_{21}}{E_{22}}) & \defeq d_E(E_{11},E_{21}) + d_E(E_{12},E_{22}) \\
    d_E(E_1,E_2) & \defeq \infty~\text{otherwise} \\
    \end{align*}
  \end{minipage}
  
  \noindent\hrulefill
  
  \begin{minipage}{1.0\textwidth}
  \begin{align*}
    d_L(\ctop,\ctop) & \defeq 0 \\
    d_L(\cneg{L_1},\cneg{L_2}) & \defeq d_L(L_1,L_2) \\
    d_L(\cconj{L_{11}}{L_{12}},\cconj{L_{21}}{L_{22}}) & \defeq d_L(L_{11},L_{21}) + d_L(L_{12},L_{22}) \\
    d_L(\cle{E_{11}}{E_{12}},\cle{E_{21}}{E_{22}}) & \defeq d_E(E_{11},E_{21}) + d_E(E_{12},E_{22}) \\
    d_L(L_1,L_2) & \defeq \infty~\text{otherwise} \\
  \end{align*}
  \end{minipage}
  
  \noindent\hrulefill
  
  \begin{minipage}{1.0\textwidth}
    \begin{align*}
      d_D(\kw{uniform}(a_1,b_1),\kw{uniform}(a_2,b_2)) & \defeq |a_1-a_2| + |b_1-b_2| \\
      d_D(D_1,D_2) & \defeq \infty~\text{otherwise} \\
    \end{align*}
  \end{minipage}
  
  \noindent\hrulefill
  
  \begin{minipage}{1.0\textwidth}
  \begin{align*}
    d_S(\iskip,\iskip) & \defeq 0 \\
    d_S(\itick{c_1},\itick{c_2}) & \defeq |c_1 - c_2| \\
    d_S(\iassign{x}{E_1},\iassign{x}{E_2}) & \defeq d_E(E_1,E_2) \\
    d_S(\isample{x}{D_1},\isample{x}{D_2}) & \defeq d_D(D_1,D_2) \\
    d_S(\iinvoke{f},\iinvoke{f} ) & \defeq 0 \\
    d_S(\iprob{p_1}{S_{11}}{S_{12}},\iprob{p_2}{S_{21}}{S_{22}}) & \defeq |p_1-p_2| + d_S(S_{11},S_{21}) + d_S(S_{12},S_{22}) \\
    d_S(\icond{L_1}{S_{11}}{S_{12}},\icond{L_2}{S_{21}}{S_{22}}) & \defeq d_L(L_1,L_2) + d_S(S_{11},S_{21}) + d_S(S_{12},S_{22}) \\
    d_S(\iloop{L_1}{S_1}, \iloop{L_2}{S_2}) & \defeq d_L(L_1,L_2)+d_S(S_1,S_2) \\
    d_S(S_{11};S_{12} , S_{21};S_{22}) & \defeq d_S(S_{11},S_{21}) + d_S(S_{12},S_{22}) \\
    d_S(S_1,S_2) & \defeq \infty ~\text{otherwise} \\
  \end{align*}
  \end{minipage}
  
  \noindent\hrulefill
  
  \begin{minipage}{1.0\textwidth}
  \begin{align*}
    d_K(\kstop, \kstop ) & \defeq 0 \\
    d_K(\kloop{L_1}{S_1}{K_1},\kloop{L_2}{S_2}{K_2}) & \defeq d_L(L_1,L_2) + d_S(S_1,S_2) + d_K(K_1,K_2) \\
    d_K(\kseq{S_1}{K_1},\kseq{S_2}{K_2}) & \defeq d_S(S_1,S_2) + d_K(K_1,K_2) \\
    d_K(K_1,K_2) & \defeq \infty ~\text{otherwise} \\
  \end{align*}
  \end{minipage}
  
  \noindent\hrulefill
  
  \caption{Metrics for expressions, conditions, distributions, statements, and continuations.}
  \label{Fi:MetricsForSemantics}
\end{figure*}

The evaluation relation $\mapsto$ can be interpreted as a \emph{distribution transformer}, i.e., a probability kernel.

\begin{lemma}\label{Lem:SoundnessForExpCond}
  Let $\gamma :\mathsf{VID} \to \bbR$ be a valuation.
  \begin{itemize}
    \item Let $E$ be an expression. Then there exists a unique $r \in \bbR$ such that $\gamma \vdash E \Downarrow r$.
    \item Let $L$ be a condition. Then there exists a unique $b \in \{ \top, \bot \}$ such that $\gamma \vdash L \Downarrow b$.
  \end{itemize}
\end{lemma}
\begin{proof}
  By induction on the structure of $E$ and $L$.
\end{proof}

\begin{lemma}\label{Lem:EvaluationUniqueness}
  For every configuration $\sigma \in \Sigma$, there exists a unique $\mu\in\bbD(\Sigma,\calO)$ such that $\sigma \mapsto \mu$.
\end{lemma}
\begin{proof}
  Let $\sigma = \tuple{\gamma,S,K,\alpha}$.
  Then by case analysis on the structure of $S$, followed by a case analysis on the structure of $K$ if $S = \iskip$.
  The rest of the proof appeals to \cref{Lem:SoundnessForExpCond}. 
\end{proof}

\begin{theorem}\label{The:EvaluationIsKernel}
  The evaluation relation $\mapsto$ defines a probability kernel in $\bbK(\Sigma,\calO)$.
\end{theorem}
\begin{proof}
  \cref{Lem:EvaluationUniqueness} tells us that $\mapsto$ can be seen as a function $\hat{\mapsto}$ defined as follows:
  \[
  \hat{\mapsto}(\sigma,A) \defeq \mu(A) \quad \text{where} \quad \sigma \mapsto \mu.
  \]
  It is clear that $\lambda A. \hat{\mapsto}(\sigma,A)$ is a probability measure.
  On the other hand, to show that $\lambda\sigma. \hat{\mapsto}(\sigma,A)$ is measurable for any $A \in \calO$, we need to prove that for $B \in \calB(\bbR)$, it holds that $\scrO(A,B) \defeq (\lambda\sigma. \hat{\mapsto}(\sigma,A))^{-1}(B) \in \calO$.
  
  We introduce \emph{skeletons} of programs to separate real numbers and discrete structures.
  \begin{align*}
    \hat{S} & \Coloneqq \iskip \mid \itick{\square_\ell} \mid \iassign{x}{\hat{E}} \mid \isample{x}{\hat{D}} \mid \iinvoke{f} \mid \iloop{\hat{L}}{\hat{S}} \\
    & \mid \iprob{\square_\ell}{\hat{S}_1}{\hat{S}_2} \mid \icond{\hat{L}}{\hat{S}_1}{\hat{S}_2} \mid \hat{S}_1; \hat{S}_2 \\
    \hat{L} & \Coloneqq \ctop \mid \cneg{\hat{L}} \mid \cconj{\hat{L}_1}{\hat{L}_2} \mid \cle{\hat{E}_1}{\hat{E}_2} \\
    \hat{E} & \Coloneqq x \mid \square_\ell \mid \eadd{\hat{E}_1}{\hat{E}_2} \mid \emul{\hat{E}_1}{\hat{E}_2} \\
    \hat{D} & \Coloneqq \kw{uniform}(\square_{\ell_a},\square_{\ell_b}) \\
    \hat{K} & \Coloneqq \kstop \mid \kloop{\hat{L}}{\hat{S}}{\hat{K}} \mid \kseq{\hat{S}}{\hat{K}}
  \end{align*}
  The \emph{holes} $\square_\ell$ are placeholders for real numbers parameterized by \emph{locations} $\ell \in \m{LOC}$.
  We assume that the holes in a program structure are always pairwise distinct.
  Let $\eta : \m{LOC} \to \bbR$ be a map from holes to real numbers and $\eta(\hat{S})$ (resp., $\eta(\hat{L})$, $\eta(\hat{E})$, $\eta(\hat{D})$, $\eta(\hat{K})$) be the instantiation of a statement (resp., condition, expression, distribution, continuation) skeleton by substituting $\eta(\ell)$ for $\square_\ell$.
  One important property of skeletons is that the \emph{distance} between any concretizations of two different skeletons is always infinity with respect to the metrics in \cref{Fi:MetricsForSemantics}.
  
  Observe that
  \[
  \scrO(A,B) = \bigcup_{\hat{S},\hat{K}} \scrO(A,B) \cap \{ \tuple{\gamma,\eta(\hat{S}),\eta(\hat{K}),\alpha} \mid \text{any}~\gamma,\alpha,\eta  \}
  \]
  and that $\hat{S},\hat{K}$ are countable families of statement and continuation skeletons.
  Thus it suffices to prove that every set in the union, which we denote by $\scrO(A,B) \cap \scrC(\hat{S},\hat{K})$ later in the proof, is measurable.
  Note that $\scrC(\hat{S},\hat{K})$ itself is indeed measurable.
  Further, the skeletons $\hat{S}$ and $\hat{K}$ are able to determine the evaluation rule for all concretized configurations.
  Thus we can proceed by a case analysis on the evaluation rules.
  
  To aid the case analysis, we define a deterministic evaluation relation $\xmapsto{\text{det}}$ by getting rid of the $\delta(\cdot)$ notations in the rules in \cref{Fi:OperationalSemantics} except probabilistic ones \textsc{(E-Sample)} and \textsc{(E-Prob)}.
  Obviously, $\xmapsto{\text{det}}$ can be interpreted as a measurable function on configurations.
  
  \begin{itemize}
    \item If the evaluation rule is deterministic, then we have
    \begin{align*}
      \scrO(A,B) \cap \scrC(\hat{S},\hat{K}) & = \{ \sigma \mid \sigma \mapsto \mu, \mu(A) \in B \} \cap \scrC(\hat{S},\hat{K}) \\
       & = \{ \sigma \mid \sigma \xmapsto{\text{det}} \sigma', [\sigma' \in A] \in B \} \cap \scrC(\hat{S},\hat{K}) \\
      & =
      \begin{dcases*}
          \scrC(\hat{S},\hat{K}) & if $\{0,1\} \subseteq B$ \\
          {\xmapsto{\text{det}}}^{-1}(A) \cap \scrC(\hat{S},\hat{K}) & if $1 \in B$ and $0 \not\in B$ \\
          {\xmapsto{\text{det}}}^{-1}(A^c) \cap \scrC(\hat{S},\hat{K}) & if $0 \in B$ and $1 \not\in B$ \\
          \emptyset & if $\{0,1\} \cap B = \emptyset$.
      \end{dcases*} 
    \end{align*}
    The sets in all the cases are measurable, so is the set $\scrO(A,B) \cap \scrC(\hat{S},\hat{K})$.
    
    \item \textsc{(E-Prob)}: Consider $B$ with the form $(-\infty,t]$ with $t \in \bbR$.
    If $t \ge 1$, then $\scrO(A,B) = \Sigma$.
    Otherwise, let us assume $t < 1$.
    Let $\hat{S} = \iprob{\square}{\hat{S}_1}{\hat{S}_2}$.
    Then we have
    \begin{align*}
      \scrO(A,B) \cap \scrC(\hat{S},\hat{K}) & = \{ \sigma \mid \sigma \mapsto \mu, \mu(A) \in B \} \cap \scrC(\hat{S},\hat{K}) \\
       &= \{ \sigma \mid \sigma \mapsto p \cdot \delta(\sigma_1) + (1-p) \cdot \delta(\sigma_2), p \cdot [\sigma_1 \in A] + (1-p) \cdot [\sigma_2 \in A] \in B \} \cap \scrC(\hat{S},\hat{K}) \\
      & = \scrC(\hat{S},\hat{K}) \cap 
       \{  \tuple{\gamma,\iprob{p}{S_1}{S_2},K,\alpha} \mid \\
      & \quad p \cdot [\tuple{\gamma,S_1,K,\alpha} \in A] + (1-p) \cdot [\tuple{\gamma,S_2,K,\alpha} \in A] \le t  \}  \\
      & = \scrC(\hat{S},\hat{K}) \cap {} \\
      & \quad \{ \tuple{\gamma,\iprob{p}{S_1}{S_2},K,\alpha} \mid \\
      & \qquad (\tuple{\gamma,S_1,K,\alpha} \in A, \tuple{\gamma,S_2,K,\alpha}\not\in A, p \le t)  \vee  (\tuple{\gamma,S_2,K,\alpha} \in A, \tuple{\gamma,S_1,K,\alpha}\not\in A, 1 - p \le t) \}. 
    \end{align*}
    The set above is measurable because $A$ and $A^c$ are measurable, as well as $\{ p \le t\}$ and $\{ p \ge 1 - t \}$ are measurable in $\bbR$.
    
    \item \textsc{(E-Sample)}: Consider $B$ with the form $(-\infty,t]$ with $t \in \bbR$.
    Similar to the previous case, we assume that $t < 1$.
    Let $\hat{S} = \isample{x}{\kw{uniform}(\square_{\ell_a},\square_{\ell_b})}$, without loss of generality.
    Then we have
    \begin{align*}
      \scrO(A,B) \cap \scrC(\hat{S},\hat{K}) & = \{\sigma \mid \sigma \mapsto \mu, \mu(A) \in B \} \cap \scrC(\hat{S},\hat{K}) \\
      & = \{ \sigma \mid \sigma \mapsto \mu_D \bind \kappa_{\sigma}, \int \kappa_{\sigma}(r)(A) \mu_D(dr) \le t \} \cap \scrC(\hat{S},\hat{K}) \\
      & = \scrC(\hat{S},\hat{K}) \cap  \{ \sigma \mid \sigma \mapsto \mu_D \bind \kappa_{\sigma},  \mu_D(\{ r \mid \tuple{\gamma[x \mapsto r], \iskip,K,\alpha} \in A \}) \le t\} \\
      & = \scrC(\hat{S},\hat{K}) \cap \{ \tuple{\gamma,\isample{x}{\kw{uniform}(a,b)},K,\alpha} \mid a<b, \\
      & \quad \mu_{\kw{uniform}(a,b)}(\{ r \mid \tuple{\gamma[x \mapsto r], \iskip,K,\alpha} \in A \}) \le t  \}, 
    \end{align*}
    where $\kappa_{\tuple{\gamma, S ,K,\alpha}} \defeq \lambda r. \delta(\tuple{\gamma[x \mapsto r], \iskip, K, \alpha})$.
    For fixed $\gamma,K,\alpha$, the set $\{r \mid \tuple{\gamma[x \mapsto r],\iskip,K,\alpha} \in A\}$ is measurable in $\bbR$.
    For the distributions considered in this article, there is a probability kernel $\kappa_D : \bbR^{\mathrm{ar}(D)} \rightsquigarrow \bbR$.
    For example, $\kappa_{\kw{uniform}}(a,b)$ is defined to be $\mu_{\kw{uniform}(a,b)}$ if $a < b$, or $\mathbf{0}$ otherwise.
    Therefore, $\lambda(a,b). \kappa_{\kw{uniform}}(a,b)(\{r \mid \tuple{\gamma[x \mapsto r],\iskip,K,\alpha} \in A\})$ is measurable, and its inversion on $(-\infty,t]$ is a measurable set on distribution parameters $(a,b)$.
    Hence the set above is measurable.
  \end{itemize}
\end{proof}

\section{Trace-Based Expected-Cost Analysis}
\label{Sec:ExpectedCostAnalysis}

In this article, we harness Markov-chain-based reasoning~\cite{ESOP:KKM16,LICS:OKK16} to develop a Markov-chain cost semantics for \lang{}, based on the evaluation relation $\sigma \mapsto \mu$.
An advantage of this approach is that it allows us to study how the cost of every single evaluation step contributes to the accumulated cost at the exit of the program.

\subsection{Preliminaries: Martingale Theory}
\label{Se:MartingaleTheory}

If $\mu_1$ and $\mu_2$ are two probability measures on $(S,\calS)$ and $(T,\calT)$, respectively, then there exists a unique probability measure $\mu$ on $(S,\calS) \otimes (T,\calT)$, by Fubini's theorem, called the \emph{product measure} of $\mu_1$ and $\mu_2$, written $\mu_1 \otimes \mu_2$, such that $\mu(A \times B) = \mu_1(A) \cdot \mu_2(B)$ for all $A \in \calS$ and $B \in \calT$.

If $\mu$ is a probability measure on $(S,\calS)$ and $\kappa : (S,\calS) \rightsquigarrow (T,\calT)$ is a probability kernel, then we can construct a probability measure on $(S,\calS) \otimes (T,\calT)$ that captures all transitions from $\mu$ through $\kappa$: $\mu \otimes \kappa \defeq \lambda(A,B). \int_A \kappa(x,B)\mu(dx)$.
If $\mu$ is a probability measure on $(S_0,\calS_0)$ and $\kappa_i : (S_{i-1},\calS_{i-1}) \rightsquigarrow (S_i,\calS_i)$ is a probability kernel for $i=1,\cdots,n$, where $n \in \bbZ^+$, then we can construct a probability measure on $\bigotimes_{i=0}^n (S_i,\calS_i)$, i.e., the space of sequences of $n$ transitions by iteratively applying the kernels to $\mu$:
\begin{align*}
  \mu \otimes \textstyle\bigotimes_{i=1}^0 \kappa_i & \defeq \mu, \\
  \mu \otimes \textstyle\bigotimes_{i=1}^{k+1} \kappa_i & \defeq (\mu \otimes \textstyle\bigotimes_{i=1}^k \kappa_i) \otimes \kappa_{k+1}, & 0 \le k < n.
\end{align*}

Let $(S_i, \calS_i), i \in \calI$ be a family of measurable spaces.
Their product, denoted by $\bigotimes_{i \in \calI} (S_i,\calS_i) = (\prod_{i \in \calI} S_i, \bigotimes_{i \in \calI} \calS_i)$, is the product space with the smallest $\sigma$-algebra such that for each $i \in \calI$, the coordinate map $\pi_i$ is measurable.
The theorem below is widely used to construct a probability measure on an infinite product via probability kernels.

\begin{theorem}[Ionescu-Tulcea]\label{The:IonescuTulcea}
  Let $(S_i,\calS_i), i \in \bbZ^+$ be a sequence of measurable spaces.
  Let $\mu_0$ be a probability measure on $(S_0,\calS_0)$.
  For each $i \in \bbN$, let $\kappa_i : \bigotimes_{k=0}^{i-1} (S_k,\calS_k) \rightsquigarrow (S_i,\calS_i)$ be a probability kernel.
  Then there exists a sequence of probability measures $\mu_i \defeq \mu_0 \otimes \bigotimes_{k=1}^i \kappa_k$, $i \in \bbZ^+$, and there exists a uniquely defined probability measure $\mu$ on $\bigotimes_{k=0}^\infty (S_k,\calS_k)$ such that $\mu_i(A) = \mu(A \times \prod_{k=i+1}^\infty S_k)$ for all $i \in \bbZ^+$ and $A \in \bigotimes_{k=0}^i \calS_i$.
\end{theorem}

Let $(\Omega,\calF,\bbP)$ be a probability space, i.e., a measure space where $\bbP$ is a probability measure.
We write $\calL^1$ for $\calL^1(\Omega,\calF,\bbP)$.
Let $X \in \calL^1$ be a random variable, and $\calG \subseteq \calF$ be a sub-$\sigma$-algebra of $\calF$.
Then there exists a random variable $Y \in \calL^1$ such that
(i) $Y$ is $\calG$-measurable, and
(ii) for each $G \in \calG$, it holds that $\expe[Y;G] = \expe[X;G]$.
Such a random variable $Y$ is said to be a version of \emph{conditional expectation} $\expe[X \mid \calG]$ of $X$ given $\calG$.
Conditional expectations admit almost-sure uniqueness.
On the other hand, if $X$ is a nonnegative measurable function, the conditional expectation $\expe[X \mid \calG]$ is also defined and almost-surely unique.

Intuitively, for some $\omega \in \Omega$, $\expe[X \mid \calG](\omega)$ is the expectation of $X$ given the set of values $Z(\omega)$ for every $\calG$-measurable random variable $Z$.
For example, if $\calG = \{\emptyset,\Omega\}$, which contains no information, then $\expe[X \mid \calG](\omega) = \expe[X]$.
If $\calG = \calF$, which contains full information, then $\expe[X \mid \calG](\omega) = X(\omega)$.

We review some useful properties of conditional expectations below.

\begin{theorem}\label{The:ConditionalExpectationProperties}
  Let $X \in \calL^1$, and $\calG,\calH$ be sub-$\sigma$-algebras of $\calF$.
  \begin{enumerate}[(1)]
    \item\label{Item:TotalExpectation} If $Y$ is a version of $\expe[X \mid \calG]$, then $\expe[Y] = \expe[X]$.
    
    \item\label{Item:Observed} If $X$ is $\calG$-measurable, then $\expe[X \mid \calG] = X$, a.s.
    
    \item\label{Item:TakingOutKnown} If $Z$ is $\calG$-measurable and $Z \cdot X \in \calL^1$, then $\expe[Z \cdot X \mid \calG] = Z \cdot \expe[X \mid \calG]$, a.s.
    Further, if $X$ is a nonnegative measurable function and $Z$ is a nonnegative $\calG$-measurable function, then the property also holds, with the understanding that both sides might be infinite.
  \end{enumerate}
\end{theorem}

A \emph{filtration} is a sequence $\{\calF_n\}_{n \in \bbZ^+}$ of sub-$\sigma$-algebras of a $\sigma$-algebra $\calF$ such that $\calF_n \subseteq \calF_{n+1}$ for all $n \in \bbZ^+$.
A probability space with a filtration, written $(\Omega,\calF,\{\calF_n\}_{n \in \bbZ^+},\bbP)$, is called a \emph{filtered probability space}.
A stochastic process $\{X_n\}_{n \in \bbZ^+}$ is said to be \emph{adapted} with respect to the filtration $\{\calF_n\}_{n \in \bbZ^+}$ if $X_n$ is $\calF_n$-measurable for each $n \in \bbZ^+$.
The most common way of producing filtrations is to generate them from stochastic processes, i.e., for a stochastic process $\{X_n\}_{n \in \bbZ^+}$, the filtration $\{\calF_n^X\}_{n \in \bbZ^+}$, given by $\calF_n^X \defeq \sigma(X_0,X_1,\cdots,X_n), n \in \bbZ^+$, is said to be \emph{generated} by $\{X_n\}_{n \in \bbZ^+}$.
A stochastic process $\{Y_n\}_{n \in \bbZ^+}$ is called an \emph{$\{\calF_n\}_{n\in\bbZ^+}$-martingale} if
(i) $\{Y_n\}_{n \in \bbZ^+}$ is $\{\calF_n\}_{n \in \bbZ^+}$-adapted,
(ii) $Y_n \in \calL^1$ for all $n \in \bbZ^+$, and
(iii) $\expe[Y_{n+1} \mid \calF_n] = Y_n$, a.s., for all $n \in \bbZ^+$.
A \emph{supermartingale} (resp., \emph{submartingale}) is defined similarly, except that (iii) is replaced by $\expe[Y_{n+1} \mid \calF_n] \le Y_n$, a.s. (resp., $\expe[Y_{n+1} \mid \calF_n] \ge Y_n$, a.s.).

Some results about martingales can be transferred to super-/sub-martingales using the following theorem.

\begin{theorem}[Doob-Meyer decomposition]\label{The:MartingaleDecomposition}
  Let $\{Y_n\}_{n \in \bbZ^+}$ be a submartingale.
  Then, there exist a martingale $\{M_n\}_{n \in \bbZ^+}$ and a predictable process $\{A_n\}_{n \in \bbN}$ (with $A_0=0$ adjoined) such that $A_n \in \calL^1$, $A_n \le A_{n+1}$, a.s., for all $n \in \bbZ^+$, and $Y_{n} = M_n + A_n, \forall n \in \bbZ^+$.
  On the other hand, if $\{Y_n\}_{n \in \bbZ^+}$ is a supermartingale, the symmetric statement with $Y_n=M_n-A_n, \forall n \in \bbZ^+$ also holds.
\end{theorem}

\subsection{A Markov-Chain Semantics}
\label{Sec:MarkovChainConstruction}

Let $(\Omega,\calF) \defeq \bigotimes_{n=0}^\infty (\Sigma,\calO)$ be a measurable space of \emph{infinite} traces on program configurations.
Let $\{\calF_n\}_{n \in \bbZ^+}$ be a filtration generated by coordinate maps $X_n(\omega) \defeq \omega_n$ for $n \in \bbZ^+$.
Let $\tuple{\scrD,S_{\mathsf{main}}}$ be an \lang{} program.
Let $\mu_0 \defeq \delta(\tuple{ \lambda\_.0 ,S_{\mathsf{main}},\kstop ,0 })$ be the initial distribution.
Let $\prob$ be the probability measure on infinite traces induced by \cref{The:IonescuTulcea} and \cref{The:EvaluationIsKernel}.
Then $(\Omega,\calF,\prob)$ forms a probability space on infinite traces of the program.
Intuitively, $\prob$ specifies the probability distribution over all possible executions of a probabilistic program.
The probability of an assertion $\theta$ with respect to $\prob$, written $\prob[\theta]$, is defined as $\prob(\{ \omega \mid \theta(\omega)~\text{is true} \})$.

To formulate the accumulated cost at the exit of the program, we define the \emph{stopping time} $T:\Omega \to \bbZ^+ \cup \{\infty\}$ of a probabilistic program as a random variable on the probability space $(\Omega,\calF,\prob)$ of program traces:
\[
T(\omega) \defeq \inf \{ n \in \bbZ^+ \mid \omega_n = \tuple{\_,\iskip,\kstop,\_} \},
\]
i.e., $T(\omega)$ is the number of evaluation steps before the trace $\omega$ reaches some termination configuration $\tuple{\_,\iskip,\kstop,\_}$.
We define the accumulated cost $A_T : \Omega \to \bbR$ with respect to the stopping time $T$ as
\[
A_T(\omega) \defeq A_{T(\omega)}(\omega),
\]
where $A_n : \Omega \to \bbR$ captures the accumulated cost at the $n$-th evaluation step for $n \in \bbZ^+$, which is defined as
\[
A_n(\omega) \defeq \alpha_n ~\text{where}~\omega_n =\tuple{\_,\_,\_,\alpha_n}.
\]
If the random variable $A_T$ is almost-surely well-defined,
the expected accumulated cost is given by the expectation $\expe[A_T]$ with respect to $\prob$.

\subsection{The Potential Method for Expected-Cost Analysis}
\label{Sec:PotentialMethod}

\paragraph{An informal account.}

\begin{figure}
  \centering
  \begin{subfigure}{0.49\columnwidth}
  \centering
  \begin{pseudo}
    \kw{func} $\mathsf{rdwalk}$() \kw{begin} \\+
      ${ \color{DarkBlue} \{ \; 2(d-x)+4 \;\} }$ \\
      \kw{if} $\cbin{<}{x}{d}$ \kw{then} \\+
        ${\color{DarkBlue} \{\;  2(d-x)+4 \;\}  }$ \\
        $\isample{t}{ \kw{uniform}({-1},2)}$; \\
        ${\color{DarkBlue} \{\;  2(d - x - t) + 5  \;\} }$ \\
        $ \iassign{x}{  \eadd{x}{t} }$; \\
        ${\color{DarkBlue} \{\;  2(d - x) + 5 \;\} }$ \\
        \kw{call} $\mathsf{rdwalk}$; \\
        ${\color{DarkBlue} \{ \; 1 \; \} }$ \\
        \kw{tick}(1) \\
        ${\color{DarkBlue} \{ \; 0 \; \} } $ \\-
      \kw{fi} \\-
    \kw{end}
  \end{pseudo}
  \end{subfigure}
  \begin{subfigure}{0.49\columnwidth}
  \centering
  \begin{pseudo}
    {\color{gray} \# $\mathsf{XID} \defeq \{x,d,t\}$ } \\
    {\color{gray} \# $\mathsf{FID} \defeq \{\mathsf{rdwalk} \}$ } \\
    {\color{gray} \# \emph{pre-condition}: $\{d > 0 \}$ } \\
    \kw{func} $\mathsf{main}$() \kw{begin} \\+
      $\iassign{x}{0}$; \\
      \kw{call} $\mathsf{rdwalk}$ \\-
    \kw{end}
  \end{pseudo}
  \end{subfigure}
  \caption{A bounded, biased random walk, implemented using recursion. The annotations show the derivation of an \emph{upper} bound on the expected accumulated cost.}
  \label{Fi:RecursiveRandomWalk}
\end{figure}

We present an informal review of the \emph{expected-potential method}~\cite{PLDI:NCH18}, which is also known as \emph{ranking super-martingales}~\cite{CAV:CS13,PLDI:WFG19,TACAS:KUH19,CAV:CFG16}, for expected-cost bound analysis of probabilistic programs.

The classic \emph{potential method} of amortized analysis~\cite{JADM:Tarjan85} can be automated to derive symbolic cost bounds for non-probabilistic programs~\cite{POPL:HJ03,APLAS:HH10}.
The basic idea is to define a \emph{potential function} $\phi: \Sigma \to \bbR^+$ that maps program states $\sigma \in \Sigma$ to nonnegative numbers,
where we assume each state $\sigma$ contains a cost-accumulator component $\sigma.\alpha$.
If a program executes with initial state $\sigma$ to final state $\sigma'$,
then it holds that $\phi(\sigma) \ge (\sigma'.\alpha - \sigma.\alpha) + \phi(\sigma')$,
where $(\sigma'.\alpha - \sigma.\alpha)$ describes the accumulated cost from $\sigma$ to $\sigma'$.
In other words, the initial potential $\phi(\sigma)$ must be sufficient to pay for the cost of the computation and for the potential $\phi(\sigma')$ of the resulting state $\sigma'$.
The potential method also enables \emph{compositional} reasoning: if a statement $S_1$ executes from $\sigma$ to $\sigma'$ and a statement $S_2$ executes from $\sigma'$ to $\sigma''$, then we have $\phi(\sigma) \ge (\sigma'.\alpha - \sigma.\alpha) + \phi(\sigma')$ and $\phi(\sigma') \ge (\sigma''.\alpha - \sigma'.\alpha) + \phi(\sigma'')$; therefore, we derive $\phi(\sigma) \ge (\sigma''.\alpha - \sigma.\alpha) + \phi(\sigma'')$ for the sequential composition $S_1;S_2$.
For non-probabilistic programs, the initial potential provides an \emph{upper} bound on the accumulated cost.
The potential method can also be used to derive \emph{lower} bounds; one only needs to change the direction of the potential inequality from $\phi(\sigma) \ge \Delta\alpha + \phi(\sigma')$ to $\phi(\sigma) \le \Delta\alpha + \phi(\sigma')$~\cite{SP:NDF17}.

This approach has been adapted to reason about expected costs of probabilistic programs~\cite{PLDI:NCH18,PLDI:WFG19}.
To derive upper bounds on the \emph{expected} accumulated cost of a program $S$ with initial state $\sigma$, one needs to take into consideration the \emph{distribution} of all possible executions.
More precisely, the potential function should satisfy the following property:
\begin{equation}\label{Eq:ExpectedPotentialInequality}
\phi(\sigma) \ge \expe_{\sigma' \sim \interp{S}(\sigma)}[ C(\sigma,\sigma') + \phi(\sigma')  ],
\end{equation}
where the notation $\expe_{x \sim \mu}[f(x)]$ represents the expected value of $f(x)$, where $x$ is drawn from the distribution $\mu$, $\interp{S}(\sigma)$ is the distribution over final states of executing $S$ from $\sigma$, and $C(\sigma,\sigma') \defeq \sigma'.\alpha - \sigma.\alpha$ is the execution cost from $\sigma$ to $\sigma'$.

\begin{example}\label{Exa:FstMomRecursiveRandomWalk}
  The program in \cref{Fi:RecursiveRandomWalk} implements a bounded, biased random walk.
  The main function consists of a single statement ``$\iinvoke{\mathsf{rdwalk}}$'' that invokes a recursive function.
  The variables $x$ and $d$ represent the current position and the ending position of the random walk, respectively.
  We assume that $d>0$ holds initially.
  In each step, the program samples the length of the current move from a uniform distribution on the interval $[{-1},2]$.
  The statement $\itick{1}$ adds one to a cost accumulator that counts the number of steps before
  the random walk ends. 
  We denote this accumulator by \id{tick} in the rest of this section.
  The program terminates with probability one and its expected accumulated cost is bounded by
  $2d+4$.

  \cref{Fi:RecursiveRandomWalk} annotates the $\mathsf{rdwalk}$ function with the derivation of an upper bound on the expected accumulated cost.
  The annotations, taken together, define an expected-potential function $\phi : \Sigma \to \bbR^+$ where a program state $\sigma \in \Sigma$ consists of a program point and a valuation for program variables.
  To justify the upper bound $2(d-x)+4$ for the function $\mathsf{rdwalk}$, one has to show that the potential right before the $\itick{1}$ statement should be at least $1$.
  This property is established by \emph{backward} reasoning on the function body:
  \begin{itemize}
    \item For $\iinvoke{\mathsf{rdwalk}}$, we apply the \emph{induction hypothesis} that the expected cost of the function $\mathsf{rdwalk}$ can be upper-bounded by $2(d-x)+4$.
    Adding the $1$ unit of potential need by the tick statement, we obtain $2(d-x)+5$ as the pre-annotation of the function call.
    
    \item For $\iassign{x}{\eadd{x}{t}}$, we substitute $x$ with $\eadd{x}{t}$ in the post-annotation of this statement to obtain the pre-annotation.
    
    \item For $\isample{t}{\kw{uniform}({-1},2)}$, because its post-annotation is $2(d-x-t)+5$, we compute its pre-annotation as
    \begin{center}
    $
    \begin{aligned}
     \expe_{t \sim \kw{uniform}({-1},2)}[2(d-x-t)+5] & = 2(d-x)+5 - 2 \cdot \expe_{t \sim \kw{uniform}({-1},2)}[t] \\
    & = 2(d-x)+5-2 \cdot \sfrac{1}{2} = 2(d-x)+4, 
    \end{aligned}
    $
    \end{center}
    which is exactly the upper bound we want to justify.
  \end{itemize}
\end{example}

In \cref{Sec:SoundDerivation}, we will present a program logic with the potential annotations shown in \cref{Fi:RecursiveRandomWalk}.
For now, we simply assume that the potential functions are defined directly on program states.

\paragraph{Trace-based reasoning.}

We review how prior work connects the expected-potential functions with martingales.
In this section, we focus on \emph{upper} bounds on the expected accumulated cost.

With the Markov-chain semantics, the expected-potential function can be defined as a measurable map from $\Sigma$ to $\bbR$ such that $\phi(\sigma)$ is an upper bound on the expected accumulated cost of the computation that \emph{continues} from the configuration $\sigma$, or more formally, $\phi(\tuple{\_,\iskip,\kstop,\_})=0$ and for any program configuration $\sigma\in\Sigma$, it holds that
\begin{equation}\label{Eq:PotentialProperty}
  \phi(\sigma) \ge \expe_{\sigma' \sim {\mapsto}(\sigma)}[ (\alpha'-\alpha) + \phi(\sigma')],
\end{equation}
where $\sigma=\tuple{\_,\_,\_,\alpha}$, $\sigma'=\tuple{\_,\_,\_,\alpha'}$, and ${\mapsto}(\sigma)$ is the probability measure after one evaluation step from the configuration $\sigma$.
Intuitively, the sum of the accumulated cost to reach a configuration $\sigma$ and the expected-potential function at $\sigma$ should be \emph{non-increasing}, on average, when the step $n$ increases.
Let $\Phi_n(\omega) \defeq \phi(\omega_n)$ be the expected-potential function at step $n$.
We define $Y_n(\omega) \defeq A_n(\omega) + \Phi_n(\omega)$ as the sum of accumulated cost and the expected-potential function at step $n$.
Then $Y_0(\omega)=A_0(\omega)+\Phi_0(\omega)=0+\phi(\omega_0)$ is the expected-potential function at the initial configuration, i.e., $Y_0 = \Phi_0$.
Taking the stopping time into consideration, we define the stopped version for these random variables as $A_T(\omega) \defeq A_{T(\omega)}(\omega)$, $\Phi_T(\omega) \defeq \Phi_{T(\omega)}(\omega)$, $Y_T(\omega) \defeq Y_{T(\omega)}(\omega)$.
Note that $A_T,\Phi_T,Y_T$ are not guaranteed to be well-defined everywhere.
Because $\phi(\sigma)=0$ if $\sigma$ is a termination configuration, we have $\Phi_T =0$ and $Y_T = A_T$.

\begin{lemma}\label{Lem:StoppedProcessConvergeAS}
 If $\prob[T < \infty] = 1$, i.e., the program terminates \emph{almost surely}, then $A_T$ is well-defined almost surely and $\prob[\lim_{n \to \infty} A_n = A_T] = 1$.
 Further, if $\{ A_n\}_{n \in \bbZ^+}$ is pointwise non-decreasing, then $\lim_{n \to \infty} \expe[A_n] = \expe[A_T]$.
\end{lemma}
\begin{proof}
  By the property of the operational semantics, for $\omega \in \Omega$ such that $T(\omega)<\infty$, we have $A_n(\omega) = A_T(\omega)$ for all $n \ge T(\omega)$.
  Then we have
  \begin{align*}
    \prob[\lim_{n \to \infty} A_n = A_T] & = \prob( \{ \omega \mid \lim_{n \to \infty} A_n(\omega) = A_T(\omega) \}) \\
    & \ge \prob(\{ \omega \mid \lim_{n \to \infty} A_n(\omega) = A_T(\omega) \wedge T(\omega) < \infty\}) \\
    & = \prob(\{ \omega \mid A_{T(\omega)}(\omega) = A_T(\omega) \wedge T(\omega) < \infty \}) \\
    & = \prob(\{ \omega \mid T(\omega) < \infty \}) \\
    & = 1. 
  \end{align*}
  
  Now let us assume that $\{A_n\}_{n \in \bbZ^+}$ is pointwise non-decreasing.
  By the property of the operational semantics, we know that $A_0 = 0$.
  Therefore, $A_n$'s are nonnegative random variables, and their expectations $\expe[A_n]$'s are well-defined.
  By \cref{The:MON}, we have $\lim_{n \to \infty} \expe[A_n] = \expe[\lim_{n \to \infty} A_n]$.
  We then conclude by the fact that $\lim_{n \to \infty} A_n = A_T$, a.s., which we just proved.
\end{proof}

In the expected-cost bound analysis, instead of establishing a bound on $\expe[A_T]$ directly, we establish a bound on $\expe[Y_T]$, which gives us the following leverage:
\begin{equation*}
\begin{array}{|p{.97\columnwidth}|}
\hline
\textrm{We reason about $\expe[Y_n]$ and then prove that $\expe[Y_T] \le \expe[Y_0]$.}\\
\hline
\end{array}
\end{equation*}
We can prove that $\expe[Y_n] \le \expe[Y_0]$ for all $n \in \bbZ^+$.

\begin{lemma}\label{Lem:RankingFunctionMartingale}
  For all $n \in \bbZ^+$, it holds that
  \[
  \expe[Y_{n+1} \mid \calF_n] \le Y_n, \text{a.s.},
  \]
  i.e., the expectation of $Y_n$ conditioned on the execution history is an invariant for $n \in \bbZ^+$.
  As a consequence (via \cref{The:ConditionalExpectationProperties}\ref{Item:TotalExpectation}), it holds that $\expe[Y_n] \le \expe[Y_0]$.
\end{lemma}
\begin{proof}
  The stochastic processes $\{\Phi_n\}_{n \in \bbZ^+}$ and $\{A_n\}_{n \in\bbZ^+}$ are adapted to the coordinate-generated filtration $\{\calF_n\}_{n \in \bbZ^+}$ as $\Phi_n(\omega)$ and $A_n(\omega)$ depend on $\omega_n$.
  Then we have
  \begin{align*}
    \expe[Y_{n+1} \mid \calF_n](\omega) & = \expe[A_{n+1} + \Phi_{n+1} \mid \calF_n](\omega) \\
    &= \expe[(A_{n+1} - A_n) +\Phi_{n+1} + A_n \mid \calF_n](\omega) \\
    &= \expe[(A_{n+1} - A_n) + \Phi_{n+1} \mid \calF_n](\omega) + A_n(\omega) & \reason{by \cref{The:ConditionalExpectationProperties}\ref{Item:Observed}}\\
    &= \expe[(\alpha_{n+1} - \alpha_n) + \phi(\omega_{n+1}) \mid \calF_n] + A_n(\omega) \\
    &= \expe_{\sigma' \sim {\mapsto}(\omega_n)}[(\alpha'-\alpha_n) + \phi(\sigma')] + A_n(\omega) \\
    &\le \phi(\omega_n) + A_n(\omega) & \reason{by \cref{Eq:PotentialProperty}} \\
    &= \Phi_n(\omega) + A_n(\omega) \\
    &= Y_n(\omega), \text{a.s.} 
  \end{align*}
\end{proof}

Therefore, if we could show that $\expe[Y_T] = \lim_{n \to \infty} \expe[Y_n]$, then we would conclude that $\expe[A_T] = \expe[Y_T] \le \expe[Y_0] = \expe[\Phi_0]$, i.e., the expected-potential function at the initial configuration gives an upper bound on the expected accumulated cost at termination configurations.

However, the property $\expe[Y_T] = \lim_{n \to \infty} \expe[Y_n]$ does \emph{not} necessarily hold.

\begin{counterexample}\label{Exa:NotConvergingMartingale}
  Consider the following random-walk program where the expected-potential method derives an unsound upper bound on the expected accumulated cost.
  \begin{center}
  \begin{pseudo}
    $\iassign{x}{0}$; \\
    \kw{while} $\cbin{<}{x}{N}$ \kw{do} \\+
      {\color{DarkBlue} $\{ x < N; N-1 \}$ } \kw{if} \kw{prob}($\sfrac{1}{2}$) \kw{then} $\iassign{x}{\eadd{x}{1}}$; $\itick{1}$ \kw{else} $\iassign{x}{\esub{x}{1}}$; $\itick{-1}$ \kw{fi} {\color{DarkBlue} $\{ x < N+1; N-1 \}$ } \\-
    \kw{od}
  \end{pseudo}
  \end{center}
  We assume that there is a map $\rho : \bbZ^+ \to \bbZ^+$, such that $\rho(k)$ records the evaluation step at the end of the $k$-th loop iteration, with respect to the Markov-chain semantics.
  For simplicity, we define $A_k' \defeq A_{\rho(k)}$, $\Phi_k' \defeq \Phi_{\rho(k)}$, and $Y_k' \defeq Y_{\rho(k)}$ as random variables at the end of the $k$-th loop iteration.
  Intuitively, the cost accumulator should be the same as the value of $x$ at any time of the execution.
  At the termination of the program, the value of $x$ should be at $n$, thus $A_T = N$ and $\expe[A_T] = N$.
  
  In each iteration, the program increases the cost accumulator by one with probability \sfrac{1}{2}, or decreases it by one with probability \sfrac{1}{2}, so the expected value of the cost accumulator stays unchanged, i.e., $\expe[A_k']=0$ for all $k \in \bbZ^+$.
  However, because we have shown that $\Phi_k' = N-1$ is a well-defined expected-potential function, 
  we have $\expe[Y_k'] = \expe[A_k'] + \expe[\Phi_k'] = N-1$ for all $k \in \bbZ^+$.
  Thus, $\lim_{k \to \infty} \expe[Y_k'] = N- 1 \neq N = \expe[A_T] = \expe[Y_T]$, and also $\lim_{n \to \infty} \expe[Y_n] \neq \expe[Y_T]$.
\end{counterexample}

\section{Soundness of Expected-Cost Analysis}
\label{Sec:SoundnessOfExpectedCostAnalysis}

It has been shown that the convergence property $\expe[Y_T] = \lim_{n \to \infty} \expe[Y_n]$ holds if all the stepwise costs are nonnegative and the analysis only considers upper bounds~\cite{PLDI:NCH18,TACAS:KUH19}.
To handle negative costs or to derive lower bounds while ruling out unsound expected-potential functions, like the one in \cref{Exa:NotConvergingMartingale}, recent research~\cite{PLDI:WFG19,POPL:HKG20,misc:SO19,CAV:BEF16} has adapted the \emph{Optional Stopping Theorem} (OST) from probability theory:

\begin{proposition}[Doob's OST~{\cite[Thm. 10.10]{book:Williams91}}]\label{Prop:OST}
  If $\expe[|Y_n|] < \infty$ for all $n \in \bbZ^+$, then $\expe[Y_T]$ exists and $\expe[Y_T] \le \expe[Y_0]$ in each of the following situations:
  \begin{enumerate}[(a)]
    \item\label{Item:OSTBT} $T$ is bounded;
    \item\label{Item:OSTPAST} $\expe[T]<\infty$ and for some $C \ge 0$, $|Y_{n+1} -Y_n| \le C$ for all $n \in \bbZ^+$;
    \item\label{Item:OSTAST} $\bbP[T < \infty] =1$ and for some $C \ge 0$, $|Y_n| \le C$ for all $n \in \bbZ^+$.
  \end{enumerate}
\end{proposition}

Note that from \cref{Item:OSTBT} to \cref{Item:OSTAST} in \cref{Prop:OST}, the constraint on the stopping time $T$ is getting weaker while the constraint on $\{Y_n\}_{n \in \bbZ^+}$ is getting stronger.
\cref{Prop:OST}\ref{Item:OSTBT} corresponds to an analogy we will
present in \cref{Exa:NonProbabilisticCase} about non-probabilistic
programs: one does not need extra constraints on the costs or
expected-potential functions if the program terminates.

For \cref{Exa:NotConvergingMartingale}, it has been shown that $\prob[T < \infty] = 1$ but $\expe[T] = \infty$~\cite{POPL:FH15}.
The only applicable OST criterion is \cref{Prop:OST}\ref{Item:OSTAST}, which requires $\{Y_k'\}_{k \in \bbZ^+}$ to be uniformly bounded.
However, the absolute value of the accumulated cost $|A_k'|$ at the $k$-th loop iteration can be as large as $k$ if the program always executes the else-branch of the probabilistic choice.
Therefore, we cannot apply OST to this program, so we cannot use the expected-potential method to reason about this program.

During our development, we find the usage of the OST in probabilistic programs sometimes not so intuitive in terms of program semantics.
Therefore, we first present a semantic characterization of \emph{optional stopping} for probabilistic programs.

\subsection{Semantic Optional Stopping}
\label{Sec:SemanticOptionalStopping}

To simplify the presentation,
let us again consider real-valued expected-potential functions for \emph{lower} bounds on expected costs.
Recall the problem: we know that $\expe[Y_n] \ge \expe[Y_0]$ for all $n \in \bbZ^+$, and we want to establish the convergence property $\lim_{n \to \infty} \expe[Y_n] = \expe[Y_T]$ so that we can conclude $\expe[A_T] = \expe[Y_T] \ge \expe[Y_0] = \expe[\Phi_0]$, where the random variable $Y_n$ is the sum of the accumulated cost $A_n$ and the expected-potential function $\Phi_n$ at the $n$-th step of the evaluations.
To start with, let us consider the case for \emph{non}-probabilistic programs.

\begin{example}\label{Exa:NonProbabilisticCase}
  If the program is non-probabilistic, then the probability measure $\bbP$ for traces is exactly a Dirac measure for some deterministic trace $\hat{\omega}$.
  As a consequence, $\expe[Y_n] = Y_n(\hat{\omega})$ for all $n \in \bbZ^+$ and $\expe[Y_T] = Y_{T(\hat{\omega})}(\hat{\omega})$.
  If the program \emph{does} terminate, i.e., $T(\hat{\omega}) < \infty$, then $\expe[Y_T] = \expe[Y_{T(\hat{\omega})}] \ge \expe[Y_0]$ immediately,
  because we already know that $\expe[Y_n] = \expe[Y_0]$ for all finite $n$.
  
  Otherwise, if the program is non-terminating, i.e., $T(\hat{\omega}) = \infty$, then $\expe[Y_T] = Y_\infty(\hat{\omega}) = 0$ by definition.
  However, in general, we have $\expe[Y_n] \neq 0$ for all $n \in \bbZ^+$.
  Consider the following non-terminating program:
  \[
  \iloop{\ctop}{\iskip}
  \]
  and similar to former examples, we assume that $A_n$, $\Phi_n$, and $Y_n$ stand for the configuration at the end of the $n$-th loop iteration.
  Let $\phi$ be the potential function such that $\phi(\sigma)=0$ if $\sigma$ is a termination configuration, and otherwise $\phi(\sigma)=1$.
  Using this potential function, we derive that $Y_n = 1$ for all $n \in \bbZ^+$; thus, $\expe[Y_T] \not\ge \expe[Y_0]$.
\end{example}

As shown by \cref{Exa:NonProbabilisticCase}, \emph{non-termination} is the reason that the convergence property may fail to hold.
%
Probabilistic programs can exhibit a mixed behavior of termination and non-termination: even though the random-walk program in \cref{Fi:RecursiveRandomWalk} terminates with probability one, it has some traces that are non-terminating.
Because we focus on the accumulated cost \emph{at termination configurations}, non-terminating traces finally have zero contribution to the expectation $\expe[Y_T]$, but they \emph{do} affect $\expe[Y_n]$ for finite $n$.
Intuitively, to establish the convergence property, we need to put a limit on the contribution to $\expe[Y_n]$ of non-terminating traces, which should approach zero when $n$ approaches infinity.

Formally, let us fix $n \in \bbZ^+$ and reason about the difference between $Y_n$ and $Y_T$. Because $Y_T=Y_n$ if $T \le n$,
\begin{align*}
  \expe[|Y_T-Y_n|] & = \bbP[T \le n] \cdot \expe[|Y_T-Y_n| \mid T \le n] + \bbP[T > n] \cdot \expe[|Y_T-Y_n| \mid T > n] \\
  & = \bbP[T > n] \cdot \expe[|Y_T-Y_n| \mid T > n]  \\
  & = \bbP[T > n] \cdot \expe[ | (A_T-A_n) + (\Phi_T-\Phi_n) | \mid T > n ] \\
  & = \bbP[T > n] \cdot \expe[ | (A_T-A_n) - \Phi_n | \mid T > n] \\
  & = \bbP[T >n] \cdot \expe[ |\textstyle \sum_{i=n}^{T-1} C_i  - \Phi_n | \mid T > n].
\end{align*}
To establish that $\lim_{n \to \infty} \expe[Y_n] = Y_T$, or equivalently, $\lim_{n \to \infty} \expe[|Y_T-Y_n|] = 0$, by the derivation above, we propose the following principle for optional stopping of probabilistic programs, which is one major contribution of this article:
\begin{equation*}
\begin{array}{|p{.96\columnwidth}|}
\hline
\textrm{It is both sufficient and necessary to show that the product of (i) the probability that the program does not terminate, and (ii) the expected gap between the expected-potential function and the remaining cost at step $n$ approaches zero when $n$ approaches infinity.} \\
\hline
\end{array}
\end{equation*}

The three situations in the classic OST (\cref{Prop:OST}) can be interpreted as sufficient conditions for the ``product-of-probability-and-expected-gap-approach-zero'' principle.

Intuitively, the principle states that if a program makes a probabilistic choice at some point, the accumulated costs of the two possible computations should not be too \emph{different} from each other, while the allowed maximal \emph{difference} is determined by how fast the probability $\prob[T > n]$ decreases as $n$ grows.

\subsection{Optional Stopping Theorems}
\label{Sec:OptionalStoppingTheorems}

In this section, we review some variants of the OST from the literature (e.g., \cref{The:OSTUICriteria}), and then present a new one (\cref{The:OSTExtended}), which has been proved practical in our development.

Let us formulate the notion of stopping times.
A random variable $T$ with values in $\bbZ^+ \cup \{\infty\}$ is called a \emph{random time}.
A random time is said to be a \emph{stopping time} with respect to the filtration $\{\calF_n\}_{n \in \bbZ^+}$ if the event $\{T \le n \}$ is $\calF_n$-measurable for all $n \in \bbZ^+$.
Let $\{Y_n\}_{n \in \bbZ^+}$ be a stochastic process and $T$ be a stopping time.
The process $\{Y_n\}_{n \in \bbZ^+}$ \emph{stopped at} $T$, denoted by $\{Y^T_n\}_{n \in \bbN}$, is defined by $Y^T_n \defeq \lambda \omega. Y_{\min(T(\omega),n)}(\omega)$.
Note that in the Markov-chain semantics, the transitions \emph{after} program termination do not mutate program states and the potential at termination states are always zero, the stochastic process $\{Y^T_n\}_{n \in \bbZ^+}$ is identical to $\{Y_n\}_{n \in \bbZ^+}$;
thus, the convergence property $\lim_{n \to \infty} \expe[Y_n] = \expe[Y_T]$ is reduced to the \emph{optional stopping} property, i.e., $\{Y^T_n\}_{n \in \bbZ^+}$ converges to $Y_T$ in $\calL^1$.
In this section, we assume that $\{Y_n\}_{n \in \bbZ^+}$ is an arbitrary martingale, i.e., it is not necessarily derived from the Markov-chain semantics.

\begin{theorem}\label{The:StoppedMartingale}
  Let $\{Y_n\}_{n \in \bbZ^+}$ be a martingale (resp., super-/sub-martingale), and let $T$ be a stopping time.
  Then the stopped process $\{Y^T_n\}_{n \in \bbZ^+}$ is also a martingale (resp., super-/sub-martingale).
\end{theorem}
\begin{proof}
  Let $C_n \defeq \mathrm{I}_{\{T \ge n\}}$, $n \in \bbN$.
  Then $\{C_n\}_{n \in \bbN}$ is predictable, and for each $n \in \bbZ^+$, we have
  \[
  \sum_{k=1}^n C_k \cdot (Y_k - Y_{k-1}) = Y_{\min(T, n)} - Y_0 = Y^T_n - Y_0.
  \]
  On the other hand, we claim that $\{\sum_{k=1}^n C_k \cdot (Y_k - Y_{k-1}) \}_{n \in \bbZ^+}$ forms a martingale null at zero.
  The process is clearly obviously $\{\calF_n\}_{n \in \bbZ^+}$-adapted.
  Because $\{C_n\}_{n \in \bbN}$ is bounded, we know that the process is in $\calL^1$.
  Finally, for each $n \in \bbZ^+$, we have
  \begin{align*}
  \expe[\sum_{k=1}^{n+1} C_k \cdot (Y_k-Y_{k-1}) \mid \calF_n] & = \sum_{k=1}^n C_k \cdot (Y_k-Y_{k-1}) + \expe[C_{n+1} \cdot (Y_{n+1}-Y_n) \mid \calF_n] \\
  & = \sum_{k=1}^n C_k \cdot (Y_k-Y_{k-1}) + C_{n+1} \cdot \expe[Y_{n+1} - Y_n \mid \calF_n] \\
  & = \sum_{k=1}^n C_k \cdot (Y_k-Y_{k-1}),
  \end{align*}
  almost surely, by \cref{The:ConditionalExpectationProperties}\ref{Item:Observed} and \ref{Item:TakingOutKnown}.
  Therefore, $\{Y^T_n\}_{n \in \bbZ^+}$ is also a martingale.
  
  For the case where $\{Y_n\}_{n \in \bbZ^+}$ is a super-/sub-martingale, we can reason similarly, except that we need the observation that $\{C_n\}_{n \in \bbN}$ is also nonnegative.
\end{proof}

We review some results about almost-sure convergence of martingales.

\begin{theorem}[Doob's martingale convergence theorem]\label{The:MartingaleConvergence}
  Let $\{Y_n\}_{n \in \bbZ^+}$ be a martingale such that $\sup_{n \in \bbZ^+} \expe[|Y_n|] < \infty$.
  Then there exists $Y \in \calL^1$ such that $\lim_{n \to \infty} Y_n = Y$, a.s.
\end{theorem}

\begin{theorem}\label{The:MartingaleConvergenceUnderPositivity}
  Let $\{Y_n\}_{n \in \bbZ^+}$ be a supermartingale (resp., submartingale) such that $\sup_{n \in \bbZ^+} \expe[Y_n^-] < \infty$ (resp., $\sup_{n \in \bbZ^+} \expe[Y_n^+] < \infty$).
  Then there exists $Y \in \calL^1$ such that $\lim_{n \to \infty} Y_n = Y$, a.s.
\end{theorem}
\begin{proof}
  We show the proof for the supermartingale case.
  Let $\{M_n\}_{n \in \bbZ^+}$ and $\{A_n\}_{n \in \bbN}$ be as in \cref{The:MartingaleDecomposition}.
  Because $M_n = Y_n + A_n \ge Y_n$, a.s., we have $\expe[M_n^-] \le \expe[Y_n^-] \le \sup_{n \in \bbZ^+} \expe[Y_n^-] < \infty$.
  On the other hand, because $\expe[M_n^+] = \expe[M_n] + \expe[M_n^-] \le \expe[M_0] + \sup_{n \in \bbN} \expe[Y_n^-] < \infty$, we have $\sup_{n \in \bbZ^+} \expe[|M_n|] < \infty$.
  By \cref{The:MartingaleConvergence}, there exists $M_\infty \in \calL^1$ such that $\lim_{n \to \infty} M_n = M_\infty$, a.s.
  
  Recall that $\{A_n\}_{n \in \bbN}$ is nonnegative and non-decreasing in $\calL^1$.
  Then by \cref{The:MON}, there exists a nonnegative random variable $A_\infty$ such that $\lim_{n \to \infty} A_n = A_\infty$, a.s., and $\expe[A_\infty] = \lim_{n \to \infty} \expe[A_n]$.
  Observe that $\lim_{n \to \infty} Y_n = \lim_{n \to \infty} (M_n - A_n) = M_\infty - A_\infty$, a.s.
  Then it suffices to show that $\expe[A_\infty] < \infty$.
  We conclude by $\expe[A_n] = \expe[M_n] - \expe[Y_n] \le \expe[M_0] + \sup_{n \in \bbZ^+} \expe[Y_n^-] < \infty$.
\end{proof}

\begin{corollary}\label{Cor:MartingaleConvergenceUnderPositivity}
  Let $\{Y_n\}_{n \in \bbZ^+}$ be a nonnegative supermartingale (or a nonpositive submartingale).
  Then there exists a random variable $Y \in \calL^1$ such that $\lim_{n \to \infty} Y_n = Y$, a.s.
\end{corollary}
\begin{proof}
  Appeal to \cref{The:MartingaleConvergenceUnderPositivity}.
\end{proof}

There is still a gap from martingale convergence to the desirable $\calL^1$-convergence:
we want to establish that $\lim_{n \to \infty} \expe[Y^T_n] = \expe[Y_T]$.
Recall that a nonempty family $\calX$ of random variables is said to be \emph{uniformly integrable} (UI) if given $\varepsilon > 0$, there exists $K \in \bbR^+$ such that $ \expe[|X|; |X| > K] < \varepsilon$ for all $X\in \calX$.
We now have a stronger convergence result on UI-martingales.

\begin{theorem}[UI-martingale convergence]\label{The:UIConvergence}
  Let $\{Y_n\}_{n \in \bbZ^+}$ be a uniformly integrable martingale.
  Then $\{Y_n\}_{n \in \bbZ^+}$ converges in $\calL^1$, where the convergence also holds a.s.
  The result still holds if $\{Y_n\}_{n \in \bbZ^+}$ is a super-/sub-martingale.
\end{theorem}

\begin{theorem}[Optional stopping under UI conditions]\label{The:OSTUI}
  For a martingale $\{Y_n\}_{n \in \bbZ^+}$, let $T$ be a stopping time such that the stopped martingale $\{Y^T_n\}_{n \in \bbZ^+}$ is uniformly integrable.
  Then $Y_T = \lim_{n \to \infty} Y^T_n$ exists a.s., $Y_T \in \calL^1$, and $\expe[Y_T] = \expe[Y_0]$.
  The result still holds if $\{Y_n\}_{n \in \bbZ^+}$ is a supermartingale (resp., submartingale), except that $\expe[Y_T] \le \expe[Y_0]$ (resp., $\expe[Y_T] \ge \expe[Y_0]$).
\end{theorem}
\begin{proof}
  By \cref{The:UIConvergence}, $\{Y^T_n\}_{n \in \bbZ^+}$ converges to some random variable $Y_T$, both a.s. and in $\calL^1$.
  Thus, $Y_T = \lim_{n \to \infty} Y^T_n$ a.s., and $\expe[Y_T] = \lim_{n \to \infty} \expe[Y^T_n]$.
  On the other hand, by the definition of martingales and \cref{The:ConditionalExpectationProperties}\ref{Item:TotalExpectation}, we have $\expe[Y^T_{n+1}] = \expe[\expe(Y^T_{n + 1} \mid \calF_n) ] = \expe[Y^T_n]$ for all $n \in \bbZ^+$.
  Thus, by a simple induction, we know that $\expe[Y^T_n] = \expe[Y^T_0] = \expe[Y_0]$ for all $n \in \bbZ^+$.
\end{proof}

We can reformulate and generalize the classic OST (\cref{Prop:OST}) using UI conditions.

\begin{theorem}\label{The:OSTUICriteria}
  Let $\{Y_n\}_{n \in \bbZ^+}$ be a martingale, and let $T$ be a stopping time.
  Then $\{Y^T_n\}_{n \in \bbZ^+}$ is uniformly integrable in each of the following situations:
  \begin{enumerate}[(a)]
    \item $T$ is almost-surely bounded;
    \item $\expe[T] < \infty$ and $\expe[|Y_{n+1} - Y_n| \mid \calF_n] \le K$, a.s., on $\{T>n\}$, for all $n \in \bbZ^+$, for some $K \in \bbR^+$;
    \item $\{Y_n\}_{n \in \bbZ^+}$ is uniformly integrable.    
  \end{enumerate}
  The result still holds if $\{Y_n\}_{n \in \bbZ^+}$ is a super-/sub-martingale.
\end{theorem}

For a special class of martingales, we have another set of conditions for optional stopping.
Note that the assumptions are minimal, compared to \cref{The:OSTUICriteria}.

\begin{theorem}\label{The:OSTUnderPositivity}
  Let $\{Y_n\}_{n \in \bbZ^+}$ be a supermartingale (resp., submartingale), and let $T$ be a stopping time.
  Suppose that that family $\{Y_{\min(T, n)}^-\}_{n \in \bbZ^+}$ (resp., $\{Y_{\min(T, n)}^+\}_{n \in \bbZ^+}$) is uniformly integrable.
  Then $Y_T = \lim_{n \to \infty} Y^T_n$ is well-defined, $Y_T \in \calL^1$, and $\expe[Y_T] \le \expe[Y_0]$ (resp., $\expe[Y_T] \ge \expe[Y_0]$).
\end{theorem}
\begin{proof}
  We show the proof for the supermartingale case.
  For any $n \in \bbZ^+$, we have
  \[
  \expe[Y_{n+1}^- \mid \calF_n] = \expe[\max(0, - Y_{n+1}) \mid \calF_n] \ge \max(0, -\expe[Y_{n+1} \mid \calF_n]) \ge \max(0, -Y_n) = Y_n^-, a.s.
  \]
  Therefore, $\{Y^-_{\min(T,n)}\}_{n \in \bbZ^+}$ is a uniformly integrable submartingale.
  By \cref{The:UIConvergence}, $Y^-_T = \lim_{n \to \infty} Y^-_{\min(T,n)}$ exists a.s., $Y^-_T \in \calL^1$, and $\expe[Y_T^-] = \lim_{n \to \infty} \expe[Y^-_{\min(T,n)}]$.
  
  On the other hand, by \cref{The:MartingaleConvergenceUnderPositivity}, $Y_T = \lim_{n \to \infty} Y^T_n$ is well-defined and $Y_T \in \calL^1$.
  Therefore, $Y_T^+ = \lim_{n \to \infty} Y^+_{\min(T,n)} = \lim_{n \to \infty} (Y^T_n + Y^-_{\min(T,n)})$ is well-defined.
  Then by Fatou's lemma, we have
  \[
  \expe[Y_T^+] = \expe[\lim_{n\to\infty} (Y^T_n + Y^-_{\min(T,n)})] \le \liminf_{n \to \infty} \expe[Y^T_n + Y^-_{\min(T,n)}] = \liminf_{n \to \infty} \expe[Y^T_n] + \lim_{n \to \infty} \expe[Y^-_{\min(T,n)}] \le \expe[Y_0] + \expe[Y_T^-].
  \]
  Thus, $\expe[Y_T] = \expe[Y_T^+] - \expe[Y_T^-] \le \expe[Y_0]$.
\end{proof}

\begin{corollary}\label{Cor:OSTUnderPositivity}
  Let $\{Y_n\}_{n \in \bbZ^+}$ be a nonnegative supermartingale (resp., nonpositive submartingale), and let $T$ be a stopping time.
  Then the limit $Y_T = \lim_{n \to \infty} Y^T_n$ is well-defined, $Y_T \in  \calL^1$, and $\expe[Y_T] \le \expe[Y_0]$ (resp., $\expe[Y_T] \ge \expe[Y_0]$).
\end{corollary}
\begin{proof}
  Appeal to \cref{The:OSTUnderPositivity}.
\end{proof}

Now we present a new extension to \cref{The:OSTUICriteria}.

\begin{theorem}\label{The:OSTExtended}
  Let $\{Y_n\}_{n \in \bbZ^+}$ be a (super-/sub-)martingale, and let $T$ be a stopping time.
  Then $\{Y^T_n\}_{n \in \bbZ^+}$ is uniformly integrable if
  $\expe[p(T)] < \infty$ and $|Y_n| \le p(n)$, a.s., on $\{T \ge n\}$, for all $n \in \bbZ^+$, for some non-decreasing function $p$.
\end{theorem}
\begin{proof}
  The idea is to show that the family $\{Y^T_n\}_{n \in \bbZ^+}$ is uniformly dominated by a random variable in $\calL^1$.
  In fact, we have
  \[
  |Y^T_n| = |Y_{\min(T,n)}| \le p(\min(T,n)) \le p(T), a.s.
  \]
  We conclude by the fact that $\expe[p(T)] < \infty$.
\end{proof}

\begin{corollary}\label{Cor:OSTExtended}
  Let $\{Y_n\}_{n \in \bbZ^+}$ be a (super-/sub-)martingale, and let $T$ be a stopping time.
  Then $\{Y^T_n\}_{n \in \bbZ^+}$ is uniformly integrable if
  there exist $\ell \in \bbN$ and $C \ge 0$ such that $\expe[T^\ell] < \infty$ and for all $n \in \bbZ^+$, $|Y_n| \le C \cdot (n+1)^{\ell}$, a.s., on $\{ T \ge n \}$.
\end{corollary}
\begin{proof}
  Appeal to \cref{The:OSTExtended}, by setting $p(n) \defeq C \cdot (n+1)^\ell$ and observing that $\expe[T^\ell] < \infty$ implies $\expe[C \cdot (T+1)^\ell] < \infty$. 
\end{proof}

\subsection{A Sound Program Logic}
\label{Sec:SoundDerivation}

\begin{figure*}
\begin{mathpar}\small
  \Rule{Valid-Ctx}{ \Forall{f \in \mathrm{dom}(\Delta)}\Forall{(\Gamma;Q,\Gamma;Q') \in \Delta(f)} \Delta \vdash \{\Gamma;Q\}~\scrD(f)~\{\Gamma';Q'\} }{ \vdash \Delta }
  \and
  \Rule{Q-Skip}{ }{ \Delta \vdash \{ \Gamma;Q \}~ \iskip ~\{ \Gamma;Q \} }
  \and
  \Rule{Q-Tick}{    Q = Q' + c }{ \Delta \vdash \{ \Gamma;Q \}~\itick{c}~\{ \Gamma;Q' \} }
  \and
  \Rule{Q-Assign}{ \Gamma = [E/x]\Gamma' \\   Q = [E/x]Q' }{ \Delta \vdash \{ \Gamma;Q  \}~\iassign{x}{E}~\{ \Gamma';Q' \} }
  \and
  \Rule{Q-Sample}{ \Gamma =\Forall{x \in \mathrm{supp}(\mu_D)} \Gamma' \\  Q = \expe_{x \sim \mu_D}[Q'] }{ \Delta \vdash \{  \Gamma;Q   \}~\isample{x}{D}~\{ \Gamma' ;Q' \} }
  \and
  \Rule{Q-Loop}{ \Delta \vdash \{ \Gamma \wedge L;Q \}~S_1~\{ \Gamma;Q \}  }{ \Delta \vdash \{ \Gamma;Q \}~\iloop{L}{S_1}~\{ \Gamma \wedge \neg L; Q \} }
  \and
  \Rule{Q-Cond}{ \Delta \vdash \{ \Gamma \wedge L; Q \}~S_1~\{\Gamma';Q' \} \\ \Delta \vdash \{ \Gamma \wedge \neg L; Q \}~S_2~\{\Gamma';Q' \}  }{ \Delta \vdash \{ \Gamma;Q \}~\icond{L}{S_1}{S_2}~\{ \Gamma';Q' \} }
  \and
  \Rule{Q-Seq}{ \Delta \vdash \{\Gamma;Q \}~S_1~\{\Gamma';Q' \} \\ \Delta \vdash \{ \Gamma';Q' \}~S_2~\{\Gamma'';Q''\}  }{ \Delta \vdash \{ \Gamma;Q \}~S_1;S_2~\{\Gamma'';Q''\} }
  \and
  \Rule{Q-Call}{  (\Gamma;Q,\Gamma';Q') \in \Delta(f) \\ c \in \bbR  }{ \Delta \vdash \{ \Gamma;  Q+c \}~\iinvoke{f}~\{ \Gamma';  Q'+c \} }
  \and
  \Rule{Q-Prob}{  \Delta \vdash \{ \Gamma;Q_1 \}~S_1~\{ \Gamma';Q' \} \\ \Delta \vdash \{ \Gamma;Q_2 \} ~S_2~\{ \Gamma';Q' \} \\ Q = p \cdot Q_1 + (1-p) \cdot Q_2 }{ \Delta \vdash \{ \Gamma;Q \} ~\iprob{p}{S_1}{S_2}~\{ \Gamma'; Q' \} }
  \and
  \Rule{Q-Weaken}{ \Delta \vdash \{ \Gamma_0;Q_0 \}~S~\{ \Gamma_0';Q_0' \} \\ \Gamma \models \Gamma_0 \\ \Gamma_0' \models \Gamma' \\ \Gamma \models Q \ge Q_0 \\ \Gamma_0' \models Q_0' \ge Q'  }{ \Delta \vdash \{\Gamma;Q \}~S~\{\Gamma';Q'\} }
  \and
  \Rule{Valid-Cfg}{ \gamma \models \Gamma \\ \Delta \vdash \{\Gamma;Q\}~S~\{\Gamma';Q'\} \\ \Delta \vdash \{ \Gamma'; Q' \}~K }{ \Delta \vdash \{ \Gamma; Q \}~\tuple{\gamma,S,K,\alpha} }
  \and
  \Rule{QK-Stop}{ }{ \Delta \vdash \{ \Gamma ; Q \}~\kstop }
  \and
  \Rule{QK-Loop}{  \Delta \vdash \{ \Gamma \wedge L; Q\}~S~\{\Gamma;Q\} \\ \Delta \vdash \{ \Gamma \wedge \neg L ;Q\}~K }{ \Delta \vdash \{ \Gamma;Q \}~\kloop{L}{S}{K} }
  \and
  \Rule{QK-Seq}{ \Delta \vdash \{ \Gamma;Q\}~S~\{\Gamma';Q'\} \\ \Delta \vdash \{\Gamma';Q'\}~K }{ \Delta \vdash \{ \Gamma;Q \}~\kseq{S}{K} }
  \and
  \Rule{QK-Weaken}{ \Delta \vdash \{\Gamma';Q'\}~K \\ \Gamma \models \Gamma' \\ \Gamma \models Q \ge Q' }{ \Delta \vdash \{\Gamma;Q\}~K }
\end{mathpar}
\caption{Inference rules of the program logic.}
\label{Fi:CompleteInferenceRules}
\end{figure*}

We formalize a program logic for expected-cost analysis in a Hoare-logic style.
Again, we focus on \emph{upper} bounds, but a logic for lower bounds can be easily obtained by using different weakening rules.
The judgment has the form $\Delta \vdash \{\Gamma;Q\}~S~\{\Gamma';Q'\}$, where
$S$ is a statement, $\{\Gamma;Q\}$ is a precondition, $\{\Gamma';Q'\}$ is a postcondition,
and $\Delta$ is a context of function specifications.
The \emph{logical context} $\Gamma : (\mathsf{VID}  \to  \bbR) \to \set{\top,\bot}$ is a predicate that describes reachable states at a program point.
The logical contexts have the same meaning as in Hoare logic.
The \emph{potential annotation} $Q : (\mathsf{VID} \to \bbR) \to \bbR$ specifies a measurable map on program states that is used to define expected-potential functions.
The semantics of the triple $\{\cdot;Q\}~S~\{\cdot;Q'\}$ is that if the rest of the computation after executing $S$ has its expected cost bounded by $Q'$, then the whole computation has its expected accumulated cost bounded by $Q$.
\emph{Function specifications} are valid pairs of pre- and post-conditions for all declared functions in a program.
For each function $f$,
a valid specification $(\Gamma;Q,\Gamma';Q') \in \Delta(f)$ is justified by the judgment $\Delta \vdash \{\Gamma;Q\}~\scrD(f)~\{\Gamma';Q'\}$, where $\scrD(f)$ is the function body of $f$.
The validity of a context $\Delta$ for function specifications is then established by the validity of all specifications in $\Delta$, denoted by $\vdash \Delta$.
To allow context-sensitive analysis, a function can have multiple specifications.

\cref{Fi:CompleteInferenceRules} presents the inference rules of the program logic.
The rule \textsc{(Q-Tick)} is the only rule that deals with costs in a program.
Because the $\itick{\cdot}$ statement does not change the program state, the logical contexts in the pre- and post-condition stay the same.
The rule \textsc{(Q-Sample)} accounts for sampling statements.
Because ``$x \sim D$'' randomly assigns a value to $x$ in the support of distribution $D$, we quantify $x$ out universally from the logical context.
To compute $Q = \expe_{x \sim \mu_D}[Q']$, where $x$ is drawn from distribution $D$, we assume the moments for $D$ are well-defined and computable,
and substitute $x^i$, $i \in \bbN$ with the corresponding moments in $Q'$.
We make this assumption because every component of $Q'$ is a polynomial over program variables.
For example, if $D = \kw{uniform}({-1},2)$, we know the following facts
\[
\expe_{x \sim \mu_D}[x^0] = 1, \expe_{x \sim \mu_D}[x^1] = \sfrac{1}{2},\expe_{x \sim \mu_D}[x^2] = 1,\expe_{x \sim \mu_D}[x^3] = \sfrac{5}{4}.
\]
Then for $Q' = xy^2+x^3y$, by the linearity of expectations, we compute $Q=\expe_{x \sim \mu_D}[Q']$ as follows:
\begin{center}
$
\begin{aligned}
  \expe_{x \sim \mu_D}[Q'] & = \expe_{x \sim \mu_D}[xy^2+x^3y]  \\
  & =  y^2 \expe_{x \sim \mu_D}[x] + y \expe_{x \sim \mu_D}[x^3]  \\
  & =  \sfrac{1}{2} \cdot  y^2+\sfrac{5}{4} \cdot y. 
\end{aligned}
$
\end{center}
The other probabilistic rule \textsc{(Q-Prob)} deals with probabilistic branching.
Intuitively, if expected execution costs of $S_1$ and $S_2$ are $q_1$ and $q_2$, respectively, and those of the computation after the branch statement is bounded by $Q'$,
then the expected cost for the whole computation should be bounded by a \emph{weighted average} of $(q_1 + Q')$ and $(q_2 + Q')$, with respect to the branching probability $p$.

The rule \textsc{(Q-Call)} handles function calls.
In the rule, we fetch the pre- and post-condition $Q_1,Q_1'$ for the function $f$ from the specification context $\Delta$.
We then combine it with a \emph{frame} $c \in \bbR$ of some constant potential.
The frame is used to account for the cost for the computation after the function call for most non-tail-recursive programs.

The structural rule \textsc{(Q-Weaken)} is used to strengthen the pre-condition and relax the post-condition.
The entailment relation $\Gamma \models \Gamma'$ states that the logical implication $\Gamma \implies \Gamma'$ is valid.
In terms of the upper bounds on the expected cost, if the triple $\{\cdot;Q\}~S~\{\cdot;Q'\}$ is valid, then we can safely add potential to the pre-condition $Q$ and remove potential from the post-condition $Q'$.

In addition to rules of the judgments for statements and function specifications, we also include rules for continuations and configurations that are used in the operational semantics.
A continuation $K$ is valid with a pre-condition $\{\Gamma;Q\}$, written $\Delta \vdash \{\Gamma;Q\}~K$, if $Q$ describes a bound on the expected accumulated cost of the computation represented by $K$ on the condition that the valuation before $K$ satisfies $\Gamma$.
Validity for configurations, written $\Delta \vdash \{\Gamma;Q\}~\tuple{\gamma,S,K,\alpha}$, is established by validity of the statement $S$ and the continuation $K$, as well as the requirement that the valuation $\gamma$ satisfies the pre-condition $\Gamma$.
Here $Q$ also describes an upper bound on the expected accumulated cost of the computation that continues from the configuration $\tuple{\gamma,S,K,\alpha}$.

To reduce the soundness proof to the extended OST, we construct an \emph{annotated transition kernel} from validity judgements $\vdash \Delta$ and $\Delta \vdash \{\Gamma;Q\}~S_{\mathsf{main}}~\{\Gamma';Q'\}$.

\begin{lemma}\label{Lem:AnnotatedKernel}
  Suppose $\vdash \Delta$ and $\Delta \vdash \{\Gamma;Q\}~S_{\mathsf{main}}~\{\Gamma';Q'\}$.
  An \emph{annotated program configuration} has the form $\tuple{\Gamma,Q ,\gamma,S,K,\alpha}$ such that $\Delta \vdash \{\Gamma;Q\}~\tuple{\gamma,S,K,\alpha}$.
  Then there exists a probability kernel $\kappa$ over annotated program configurations such that:
  
  For all $\sigma=\tuple{\Gamma,Q,\gamma,S,K,\alpha} \in \mathrm{dom}(\kappa)$, it holds that
  \begin{enumerate}[(i)]
    \item $\kappa$ is the same as the evaluation relation $\mapsto$ if the annotations are omitted, i.e.,
    \[
    \kappa(\sigma) \bind \lambda\tuple{\_,\_,\gamma',S',K',\alpha'}. \delta(\tuple{\gamma',S',K',\alpha'}) = {\mapsto}(\tuple{\gamma,S,K,\alpha}),
    \]
    and
    
    \item $Q(\gamma) \ge \expe_{\sigma' \sim \kappa(\sigma)} [(\alpha'-\alpha) + Q'(\gamma') ]$ where $\sigma'=\tuple{\_,Q',\gamma',\_,\_,\alpha'}$.
  \end{enumerate}
\end{lemma}

Before proving the soundness, we show that the program logic admits a \emph{relaxation} rule.

\begin{lemma}\label{Lem:TypingRelax}
  Suppose $\vdash \Delta$.
  If $\Delta \vdash \{\Gamma;Q\}~S~\{\Gamma';Q'\}$, then for all $c \in \bbR$, the judgment $\Delta \vdash \{\Gamma;Q +c\}~S~\{\Gamma';Q' +c \}$ is derivable.
\end{lemma}
\begin{proof}
  By induction on the derivation of $\Delta \vdash \{\Gamma;Q\}~S~\{\Gamma';Q'\}$.
  \begin{itemize}
    \item $\small\Rule{Q-Skip}{ }{ \Delta \vdash \{\Gamma;Q\}~\iskip~\{\Gamma;Q\} }$
    
    By \textsc{(Q-Skip)}, we immediately have $\Delta \vdash \{\Gamma;Q +c\}~\iskip~\{\Gamma;Q +c\}$.
    
    \item $\small\Rule{Q-Tick}{ Q =  Q' +c'  }{ \Delta \vdash \{\Gamma;Q\}~\itick{c'}~\{\Gamma;Q'\} }$
    
    We have $(Q' + c) + c' = (Q'+c') + c = Q + c$.
    Then we conclude by \textsc{(Q-Tick)}.
    
    \item $\small\Rule{Q-Assign}{ \Gamma = [E/x]\Gamma' \\ Q = [E/x]Q' }{ \Delta \vdash \{\Gamma;Q\}~\iassign{x}{E}~\{\Gamma';Q'\} }$
    
    Because $c \in \bbR$ is a constant, we know that $[E/x](Q' +c) = [E/x]Q' + c = Q+c$.
    Then we conclude by \textsc{(Q-Assign)}.
    
    \item $\small\Rule{Q-Sample}{ \Gamma = \Forall{x \in \mathrm{supp}(\mu_D)} \Gamma' \\ Q = \expe_{x \sim \mu_D}[Q'] }{ \Delta \vdash \{\Gamma;Q\}~\isample{x}{D}~\{\Gamma';Q'\} }$
    
    By the linearity of expectations, we know that $\expe_{x \sim \mu_D}[Q' +c] = \expe_{x \sim \mu_D}[Q'] +c = Q +c$.
    Then we conclude by \textsc{(Q-Sample)}.
    
    \item $\small\Rule{Q-Call}{ (\Gamma;Q,\Gamma';Q') \in \Delta(f) \\ d \in \bbR }{ \Delta \vdash \{\Gamma;Q +d\}~\iinvoke{f}~\{\Gamma';Q' +d \} }$
    
    Let $e \defeq d + c$.
    Then we conclude by \textsc{(Q-Call)} with the frame set as $e$.
    
    \item $\small\Rule{Q-Prob}{ \Delta \vdash \{\Gamma;Q_1\}~S_1~\{\Gamma';Q'\} \\ \Delta \vdash \{\Gamma;Q_2\}~S_2~\{\Gamma';Q'\} \\ Q = p \cdot Q_1 + (1-p) \cdot Q_2 }{ \Delta \vdash \{\Gamma;Q\}~\iprob{p}{S_1}{S_2}~\{\Gamma';Q'\} }$
    
    By induction hypothesis, we have $\Delta \vdash \{\Gamma;Q_1+c\}~S_1~\{\Gamma';Q'+c \}$ and $\Delta \vdash \{\Gamma;Q_2+ c \}~S_2~\{\Gamma';Q' +c\}$.
    Then $p \cdot (Q_1 + c) + (1-p) \cdot (Q_2+c) = p \cdot Q_1 + (1-p) \cdot Q_2 + c = Q + c$.
    Then we conclude by \textsc{(Q-Prob)}.
    
    \item $\small\Rule{Q-Cond}{ \Delta \vdash \{\Gamma \wedge L; Q\}~S_1~\{\Gamma';Q'\} \\ \Delta \vdash \{\Gamma \wedge \neg L;Q\} ~S_2~\{\Gamma';Q'\}  }{ \Delta \vdash \{\Gamma;Q\}~\icond{L}{S_1}{S_2}~\{\Gamma';Q'\} }$
    
    By induction hypothesis, we have $\Delta \vdash \{\Gamma \wedge L; Q +c \}~S_1~\{\Gamma';Q' +c\}$ and $\Delta \vdash \{\Gamma \wedge \neg L; Q+c \}~S_2~\{\Gamma';Q'+c\}$.
    Then we conclude by \textsc{(Q-Cond)}.
    
    \item $\small\Rule{Q-Loop}{ \Delta \vdash \{\Gamma \wedge L;Q\}~S~\{\Gamma;Q\} }{ \Delta \vdash \{\Gamma;Q\}~\iloop{L}{S}~\{\Gamma \wedge \neg L; Q\} }$
    
    By induction hypothesis, we have $\Delta \vdash \{\Gamma \wedge L; Q +c\}~S~\{\Gamma; Q +c \}$.
    Then we conclude by \textsc{(Q-Loop)}.
    
    \item $\small\Rule{Q-Seq}{ \Delta \vdash \{\Gamma;Q\}~S_1~\{\Gamma';Q'\} \\ \Delta \vdash \{\Gamma';Q'\}~S_2~\{\Gamma'';Q''\} }{ \Delta \vdash \{\Gamma;Q\}~S_1;S_2~\{\Gamma'';Q''\} }$
    
    By induction hypothesis, we have $\Delta \vdash \{\Gamma;Q +c\}~S_1~\{\Gamma';Q'+c\}$ and $\Delta \vdash \{\Gamma'; Q'+c\}~S_2~\{\Gamma'';Q''+c\}$.
    Then we conclude by \textsc{(Q-Seq)}.
    
    \item $\small\Rule{Q-Weaken}{ \Delta \vdash \{\Gamma_0;Q_0\}~S~\{\Gamma_0';Q_0'\} \\ \Gamma \models \Gamma_0 \\ \Gamma_0' \models \Gamma' \\ \Gamma \models Q \ge Q_0 \\ \Gamma_0' \models Q_0' \ge Q' }{ \Delta \vdash \{\Gamma;Q\}~S~\{\Gamma';Q'\} }$
    
    By induction hypothesis, we have $\Delta \vdash \{\Gamma_0;Q_0+c\}~S~\{\Gamma_0';Q_0'+c\}$.
    To apply \textsc{(Q-Weaken)}, we need to show that $\Gamma \models Q +c \ge Q_0 +c$ and $\Gamma_0' \models Q_0' +c \ge  Q' +c$.
    We conclude by the fact that $c \in \bbR$ is a constant.
  \end{itemize}
\end{proof}

Now we can construct the annotated transition kernel to reduce the soundness proof to OST.

\begin{proof}[Proof of \cref{Lem:AnnotatedKernel}]
  Let $\nu \defeq {\mapsto}(\tuple{\gamma,S,K,\alpha})$.
  By inversion on $\Delta \vdash \{\Gamma;Q\}~\tuple{\gamma,S,K,\alpha}$, we know that $\gamma \models \Gamma$, $\Delta \vdash \{\Gamma;Q\}~S~\{\Gamma';Q'\}$, and $\Delta \vdash \{\Gamma';Q'\}~K$ for some $\Gamma',Q'$.
  We construct a probability measure $\mu$ as $\kappa(\tuple{\Gamma,Q,\gamma,S,K,\alpha})$ by induction on the derivation of $\Delta \vdash \{\Gamma;Q\}~S~\{\Gamma';Q'\}$.
  
  \begin{itemize}
    \item $\small\Rule{Q-Skip}{ }{ \Delta \vdash \{ \Gamma;Q \}~\iskip~\{\Gamma;Q\} }$
    
    By induction on the derivation of $\Delta \vdash \{\Gamma;Q\}~K$.
    
    \begin{itemize}
      \item $\small\Rule{QK-Stop}{ }{ \Delta \vdash \{ \Gamma;Q \}~\kstop }$
      
      We have $\nu = \delta(\tuple{\gamma,\iskip,\kstop,\alpha})$.
      Then we set $\mu = \delta(\tuple{\Gamma,Q,\gamma,\iskip,\kstop,\alpha })$.
      
      \item $\small\Rule{QK-Loop}{ \Delta \vdash \{\Gamma \wedge L; Q\}~S~\{\Gamma;Q\} \\ \Delta \vdash \{\Gamma \wedge \neg L; Q\}~K }{ \Delta \vdash \{\Gamma;Q\}~\kloop{L}{S}{K} }$
      
      Let $b \in \{\bot,\top\}$ be such that $\gamma \vdash L \Downarrow b$.
      
      If $b = \top$, then $\nu = \delta(\tuple{\gamma,S,\kloop{L}{S}{K},\alpha})$.
      We set $\mu = \delta(\tuple{ \Gamma \wedge L,Q,\gamma,S,\kloop{L}{S}{K},\alpha })$.
      In this case, we know that $\gamma \models \Gamma \wedge L$.
      By the premise, we know that $\Delta \vdash \{\Gamma\wedge L;Q\}~S~\{\Gamma;Q\}$.
      It then remains to show that $\Delta \vdash \{ \Gamma;Q\}~\kloop{L}{S}{K}$.
      By \textsc{(QK-Loop)}, it suffices to show that $\Delta \vdash \{\Gamma \wedge L; Q\}~S~\{\Gamma;Q\}$ and $\Delta \vdash \{\Gamma \wedge \neg L; Q\}~K$.
      Then appeal to the premise.
      
      If $b = \bot$, then $\mu = \delta(\tuple{\gamma,\iskip,K,\alpha})$.
      We set $\mu = \delta(\tuple{\Gamma \wedge \neg L, Q, \gamma,\iskip,K,\alpha})$.
      In this case, we know that $\gamma \models \Gamma \wedge \neg L$.
      By \textsc{(Q-Skip)}, we have $\Delta \vdash \{\Gamma \wedge \neg L;Q\}~\iskip~\{\Gamma \wedge \neg L; Q\}$.
      It then remains to show that $\Delta \vdash \{\Gamma \wedge \neg L; Q\}~K$.
      Then appeal to the premise.
      
      In both cases, $\gamma$ and $Q$ do not change.
      
      \item $\small\Rule{QK-Seq}{ \Delta \vdash \{\Gamma;Q\}~S~\{\Gamma';Q'\} \\ \Delta \vdash \{\Gamma';Q'\}~K }{ \Delta \vdash \{\Gamma;Q\}~\kseq{S}{K} }$
      
      We have $\nu = \delta(\tuple{\gamma,S,K,\alpha})$.
      Then we set $\mu = \delta(\tuple{\Gamma,Q,\gamma,S,K,\alpha})$.
      By the premise, we know that $\Delta \vdash \{\Gamma;Q\}~S~\{\Gamma';Q'\}$ and $\Delta \vdash \{\Gamma';Q'\}~K$.
      Also $\gamma$ and $Q$ do not change.
      
      \item $\small\Rule{QK-Weaken}{ \Delta \vdash \{\Gamma';Q'\}~K \\ \Gamma \models \Gamma' \\ \Gamma \models Q \ge Q' }{ \Delta \vdash \{\Gamma;Q\}~K }$
      
      Because $\gamma \models \Gamma$ and $\Gamma \models \Gamma'$, we know that $\gamma \models \Gamma'$.
      Let $\mu'$ be obtained from the induction hypothesis on $\Delta \vdash \{\Gamma';Q'\}~K$.
      Then ${Q'}(\gamma) \ge \expe_{\sigma' \sim \mu'}[(\alpha'-\alpha) + {Q''}(\gamma')]$, where $\sigma'=\tuple{\_,Q'',\gamma',\_,\_,\alpha'}$.
      We set $\mu = \mu'$.
      Because $\Gamma \models Q \ge Q'$ and $\gamma \models \Gamma$, we conclude that $Q(\gamma) \ge {Q'}(\gamma)$.
    \end{itemize}

  \item $\small\Rule{Q-Tick}{ Q =  Q'+c }{ \Delta \vdash \{\Gamma;Q\}~\itick{c}~\{\Gamma;Q'\} }$
  
  We have $\nu = \delta(\tuple{\gamma,\iskip,K,\alpha+c})$.
  Then we set $\mu = \delta(\tuple{\Gamma,Q',\gamma,\iskip,K,\alpha+c})$.
  By \textsc{(Q-Skip)}, we have $\Delta \vdash \{\Gamma;Q'\}~\iskip~\{\Gamma;Q'\}$.
  Then by the assumption, we have $\Delta \vdash \{\Gamma;Q'\}~K$.
  It remains to show that $Q(\gamma) \ge c + {Q'}(\gamma)$.
  Indeed, we have $Q(\gamma) = c +  {Q'}(\gamma)$ by the premise.
  
  \item $\small\Rule{Q-Assign}{ \Gamma = [E/x]\Gamma' \\ Q = [E/x]Q' }{ \Delta \vdash \{\Gamma;Q\}~\iassign{x}{E}~\{\Gamma';Q'\} }$
  
  Let $r \in \bbR$ be such that $\gamma \vdash E \Downarrow r$.
  We have $\nu = \delta(\tuple{\gamma[x \mapsto r],\iskip,K,\alpha})$.
  Then we set $\mu = \delta(\tuple{\Gamma',Q',\gamma[x \mapsto r],\iskip,K,\alpha})$.
  Because $\gamma \vdash \Gamma$, i.e., $\gamma \vdash [E/x]\Gamma'$, we know that $\gamma[x \mapsto r] \vdash \Gamma'$.
  By \textsc{(Q-Skip)}, we have $\Delta \vdash \{\Gamma';Q'\}~\iskip~\{\Gamma';Q'\}$.
  Then by the assumption, we have $\Delta \vdash \{\Gamma';Q'\}~K$.
  It remains to show that $Q(\gamma) = {Q'}(\gamma[x \mapsto r])$.
  By the premise, we have $Q = [E/x]Q'$, thus $Q(\gamma) = {[E/x]Q'}(\gamma) = {Q'}(\gamma[x \mapsto r])$.
  
  \item $\small\Rule{Q-Sample}{ \Gamma = \Forall{x \in \mathrm{supp}(\mu_D)} \Gamma' \\ Q = \expe_{x \sim \mu_D}[Q'] }{ \Delta \vdash \{\Gamma;Q\}~\isample{x}{D}~\{\Gamma';Q'\} }$
  
  We have $\nu = \mu_D \bind \lambda r. \delta(\tuple{\gamma[x \mapsto r],\iskip,K,\alpha})$.
  Then we set $\mu = \mu_D \bind \lambda r. \delta(\tuple{\Gamma',Q',\gamma[x \mapsto r],\iskip,K,\alpha})$.
  For all $r \in \mathrm{supp}(\mu_D)$, because $\gamma \models \Forall{x \in \mathrm{supp}(\mu_D)} \Gamma'$, we know that $\gamma[x \mapsto r] \models \Gamma'$.
  By \textsc{(Q-Skip)}, we have $\Delta \vdash \{\Gamma';Q'\}~\iskip~\{\Gamma';Q'\}$.
  Then by the assumption, we have $\Delta \vdash \{\Gamma';Q'\}~K$.
  It remains to show that $Q(\gamma) \ge \expe_{r \sim \mu_D}[ {Q'}(\gamma[x \mapsto r])]$.
  Indeed, because $Q = \expe_{x \sim \mu_D}[Q']$, we know that $Q(\gamma) = {\expe_{x \sim \mu_D}[Q']}(\gamma) = \expe_{r \sim \mu_D}[{Q'}(\gamma[x \mapsto r])]$.
  
  \item $\small\Rule{Q-Call}{  (\Gamma;Q,\Gamma';Q') \in \Delta(f) \\ c \in \bbR }{ \Delta \vdash \{\Gamma;  Q+c \}~\iinvoke{f}~\{\Gamma';  Q' + c \} }$
  
  We have $\nu = \delta(\tuple{\gamma,\scrD(f),K,\alpha})$.
  Then we set $\mu = \delta(\tuple{\Gamma, Q+c ,\gamma,\scrD(f),K,\alpha})$.
  By the premise, we know that $\Delta \vdash \{\Gamma;Q\}~\scrD(f)~\{\Gamma';Q'\}$.
  By \cref{Lem:TypingRelax}, we know that $\Delta \vdash \{\Gamma;  Q+c \}~\scrD(f)~\{\Gamma';  Q'+c \} $.
  Also $\gamma$ and $Q+c$ do not change.
  
  \item $\small\Rule{Q-Prob}{ \Delta \vdash \{\Gamma;Q_1\}~S_1~\{\Gamma';Q'\} \\ \Delta \vdash \{\Gamma;Q_2\}~S_2~\{\Gamma';Q'\} \\ Q = p \cdot Q_1 + (1-p) \cdot Q_2 }{ \Delta \vdash \{\Gamma;Q\}~\iprob{p}{S_1}{S_2} ~\{\Gamma';Q'\} }$
  
  We have $\nu = p \cdot \delta(\tuple{\gamma,S_1,K,\alpha}) + (1-p) \cdot \delta(\tuple{\gamma,S_2,K,\alpha})$.
  Then we set $\mu = p \cdot \delta(\tuple{\Gamma,Q_1,\gamma,S_1,K,\alpha} + (1-p) \cdot \delta(\tuple{\Gamma,Q_2,\gamma,S_2,K,\alpha})$.
  From the assumption and the premise, we know that $\gamma \models \Gamma$, $\Delta \vdash \{\Gamma';Q'\}~K$, and $\Delta \vdash \{\Gamma;Q_1\}~S_1~\{\Gamma';Q'\}$, $\Delta \vdash \{\Gamma;Q_2\}~S_2~\{\Gamma';Q'\}$.
  It remains to show that $Q(\gamma) \ge (p \cdot {Q_1}(\gamma)) + ((1-p) \cdot {Q_2}(\gamma))$.
  On the other hand, from the premise, we have $Q = p \cdot Q_1 + (1-p) \cdot Q_2$.
  Therefore, we have $Q(\gamma) = p \cdot Q_1(\gamma) + (1-p) \cdot Q_2(\gamma)$.
  
  \item $\small\Rule{Q-Cond}{ \Delta \vdash \{\Gamma \wedge L;Q\}~S_1~\{\Gamma';Q'\} \\ \Delta \vdash \{\Gamma \wedge \neg L; Q\}~S_2~\{\Gamma';Q'\} }{ \Delta \vdash \{\Gamma;Q\}~\icond{L}{S_1}{S_2}~\{\Gamma';Q'\} }$
  
  Let $b \in \{\top,\bot\}$ be such that $\gamma \vdash L \Downarrow b$.
  
  If $b = \top$, then $\nu = \delta(\tuple{\gamma,S_1,K,\alpha})$.
  We set $\mu = \delta(\tuple{\Gamma \wedge L,Q,\gamma,S_1,K,\alpha})$.
  In this case, we know that $\gamma \models \Gamma \wedge L$.
  By the premise and the assumption, we know that $\Delta \vdash \{\Gamma \wedge L; Q\}~S_1~\{\Gamma';Q'\}$ and $\Delta \vdash \{\Gamma';Q'\}~K$.
  
  If $b = \bot$, then $\nu = \delta(\tuple{\gamma,S_2,K,\alpha})$.
  We set $\mu = \delta(\tuple{\Gamma \wedge \neg L,Q,\gamma,S_2,K,\alpha})$.
  In this case, we know that $\gamma \models \Gamma \wedge \neg L$.
  By the premise and the assumption, we know that $\Delta \vdash \{\Gamma \wedge \neg L; Q\}~S_2~\{\Gamma';Q'\}$ and $\Delta \vdash \{\Gamma';Q'\}~K$.
  
  In both cases, $\gamma$ and $Q$ do not change.
  
  \item $\small\Rule{Q-Loop}{ \Delta \vdash \{\Gamma \wedge L;Q\}~S~\{\Gamma;Q\} }{ \Delta \vdash \{\Gamma;Q\}~\iloop{L}{S_1}~\{\Gamma \wedge \neg L;Q \} }$
  
  We have $\nu = \delta(\tuple{\gamma,\iskip,\kloop{L}{S}{K},\alpha})$.
  Then we set $\mu = \delta(\tuple{\Gamma,Q,\gamma,\iskip,\kloop{L}{S}{K},\alpha})$.
  By \textsc{(Q-Skip)}, we have $\Delta \vdash \{\Gamma;Q\}~\iskip~\{\Gamma;Q\}$.
  Then by the assumption $\Delta \vdash \{\Gamma \wedge \neg L;Q\}~K$ and the premise, we know that $\Delta \vdash \{\Gamma;Q\}~\kloop{L}{S}{K}$ by \textsc{(QK-Loop)}.
  Also $\gamma$ and $Q$ do not change.
  
  \item $\small\Rule{Q-Seq}{ \Delta \vdash \{\Gamma;Q\}~S_1~\{\Gamma';Q'\} \\ \Delta \vdash \{\Gamma';Q'\}~S_2~\{\Gamma'';Q''\} }{ \Delta \vdash \{\Gamma;Q\}~S_1;S_2~\{\Gamma'';Q''\} }$
  
  We have $\nu = \delta(\tuple{\gamma,S_1,\kseq{S_2}{K},\alpha})$.
  Then we set $\mu = \delta(\tuple{\Gamma,Q,\gamma,S_1,\kseq{S_2}{K},\alpha})$.
  By the first premise, we have $\Delta \vdash \{\Gamma;Q\}~S_1~\{\Gamma';Q'\}$.
  By the assumption $\Delta \vdash \{\Gamma'';Q''\}~K$ and the second premise, we know that $\Delta \vdash \{\Gamma';Q'\}~\kseq{S_2}{K}$ by \textsc{(QK-Seq)}.
  Also $\gamma$ and $Q$ do not change. 
  
  \item $\small\Rule{Q-Weaken}{ \Delta \vdash \{\Gamma_0;Q_0\}~S~\{\Gamma_0';Q_0'\} \\ \Gamma \models \Gamma_0 \\ \Gamma_0' \models \Gamma' \\ \Gamma \models Q \ge Q_0 \\ \Gamma_0' \models Q_0' \ge Q' }{ \Delta \vdash \{\Gamma;Q\}~S~\{\Gamma';Q'\} }$
  
  By $\gamma \models \Gamma$ and $\Gamma \models \Gamma_0$, we know that $\gamma \models \Gamma_0$.
  By the assumption $\Delta \vdash \{\Gamma';Q'\}~K$ and the premise $\Gamma_0' \models \Gamma'$, $\Gamma_0' \models Q_0' \ge Q'$, we derive $\Delta \vdash \{\Gamma_0';Q_0'\}~K$ by \textsc{(QK-Weaken)}.
  Thus let $\mu_0$ be obtained by the induction hypothesis on $\Delta \vdash \{\Gamma_0;Q_0\}~S~\{\Gamma_0';Q_0'\}$.
  Then ${Q_0}(\gamma) \ge \expe_{\sigma' \sim \mu_0}[(\alpha'-\alpha)+ {Q''}(\gamma')]$, where $\sigma' = \tuple{\_,Q'',\gamma',\_,\_,\alpha'}$.
  We set $\mu = \mu_0$. 
  By the premise $\Gamma \models Q \ge Q_0$ and $\gamma \models \Gamma$, we conclude that ${Q}(\gamma) \ge {Q_0}(\gamma)$.
  \end{itemize}
\end{proof}

Therefore, we can use the annotated kernel $\kappa$ above to re-construct the Markov-chain semantics.
Then we can define the potential function on annotated program configurations as $\phi(\sigma) \defeq Q(\gamma)$ where $\sigma=\tuple{\_,Q,\gamma,\_,\_,\_}$.

The next step is to apply the extended OST (\cref{Cor:OSTExtended}).
Recall that the theorem requires that for some $\ell \in \bbN$ and $C \ge 0$, $|Y_n| \le C \cdot (n+1)^{\ell}$ almost surely for all $n \in \bbZ^+$.
One sufficient condition for the requirement is to 
(i) assume all the quantitative contexts $Q$ is the derivation are \emph{polynomials} over program variables up to some fixed degree $d \in \bbN$, and
(ii) assume the \emph{bounded-update} property, i.e., every (deterministic or probabilistic) assignment to a program variable updates the variable with a bounded change.
As observed by Wang et al.~\autocite{PLDI:WFG19}, bounded updates are common in practice.
We formulate the idea as follows.

\begin{lemma}\label{Lem:BoundedUpdate}
  If there exists $C_0 \ge 0$ such that for all $n \in \bbZ^+$ and $x \in \mathsf{VID}$, it holds that $\prob[\abs{\gamma_{n+1}(x)- \gamma_n(x) } \le C_0] = 1$ where $\omega$ is an infinite trace, $\omega_n = \tuple{\gamma_n,\_,\_,\_}$, and $\omega_{n+1} = \tuple{\gamma_{n+1},\_,\_,\_}$,
  then there exists $C \ge 0$ such that  for all $n \in \bbZ^+$, $|Y_n| \le C \cdot (n+1)^{d}$ almost surely.
\end{lemma}
\begin{proof}
  Let $C_1 \ge 0$ be such that for all $\itick{c}$ statements in the program, $|c| \le C_1$.
  Then for all $\omega$, if $\omega_n = \tuple{\_,\_,\_,\alpha_n}$, then $|\alpha_n| \le n \cdot C_1$.
  On the other hand, we know that $\prob[|\gamma_n(x) - \gamma_0(x)| \le C_0 \cdot n] = 1$ for any variable $x$.
  As we assume all the program variables are initialized to zero, we know that $\prob[|\gamma_n(x)| \le C_0 \cdot n] = 1$.
  From the construction in the proof of \cref{Lem:AnnotatedKernel}, we know that all the quantitative contexts (represented by polynomials over program variables) should have almost surely bounded coefficients.
  Let $C_2 \ge 0$ be such a bound.
  Also, $\Phi_n(\omega) = \phi(\omega_n) = {Q_n}(\gamma_n)$, and
  \[
  |{Q_n}(\gamma_n)_k| \le \sum_{i=0}^{d} C_2 \cdot |\mathsf{VID}|^{i} \cdot |C_0 \cdot n|^i \le C_3 \cdot (n+1)^{d}, \textrm{a.s.},
  \]
  for some sufficiently large constant $C_3$.
  Thus $|Y_n| = |A_n+\Phi_n| \le |A_n| + |\Phi_n| \le C_1 \cdot n + C_3 \cdot (n+1)^d \le C_4 \cdot (n+1)^d$, a.s.,
  for some sufficiently large constant $C_4$.
\end{proof}

Now we prove the soundness of bound inference.

\begin{theorem}\label{The:SoundnessOfLogic}
  Suppose $\Delta \vdash \{\Gamma;Q\}~S_{\mathsf{main}}~\{\Gamma';0\}$ and $\vdash \Delta$.
  Then $\expe[A_T] \le Q(\lambda\_.0)$, i.e., the expected accumulated cost $\expe[A_T]$ upon program termination is upper-bounded by $Q(\lambda\_.0)$ where $Q$ is the quantitative context and $\lambda\_.0$ is the initial valuation, if \emph{both} of the following properties hold:
  \begin{enumerate}[(i)]
    \item $\expe[T^{d}] < \infty$, and
    
    \item there exists $C_0 \ge 0$ such that for all $n \in \bbZ^+$ and $x \in \mathsf{VID}$, it holds almost surely that $|\gamma_{n+1}(x) - \gamma_n(x)| \le C_0$ where $\tuple{\gamma_n,\_,\_,\_} = \omega_n$ and $\tuple{\gamma_{n+1},\_,\_,\_} = \omega_{n+1}$ of an infinite trace $\omega$.
  \end{enumerate}
\end{theorem}
\begin{proof}
  By \cref{Lem:BoundedUpdate}, there exists $C \ge 0$ such that $|Y_n| \le C \cdot (n+1)^{d}$ almost surely for all $n \in \bbZ^+$.
  By the assumption, we also know that $\expe[T^{d}] < \infty$.
  Thus by \cref{Cor:OSTExtended,The:OSTUI}, we conclude that $\expe[Y_T] \le \expe[Y_0]$, i.e., $\expe[A_T] \le \expe[\Phi_0] = Q(\lambda\_.0)$.
\end{proof}


\printbibliography[heading=bibintoc]
\end{document}